\documentclass[%
 jmp,
 amsmath,amssymb,
preprint,%
]{revtex4-1}

\usepackage{graphicx}% Include figure files
\usepackage{dcolumn}% Align table columns on decimal point
\usepackage{bm}% bold math
\usepackage[utf8]{inputenc}
\usepackage[T1]{fontenc}
\usepackage{mathptmx}
\usepackage{amsthm, thm-restate}
\newtheorem{theorem}{Theorem}[section]

\newtheorem{lemma}[theorem]{Lemma}
\newtheorem{conjecture}{Conjecture}
\theoremstyle{definition}
\newtheorem{definition}{Definition}[section]
\usepackage{esvect}
\usepackage{amsmath}
\newcommand{\R}{\mathbb{R}}
\usepackage{physics}
\begin{document}

\preprint{AIP/123-QED}

\title[Continuum Limits of the 1D DTQW]{Continuum Limits of the 1D Discrete Time Quantum Walk}
% Force line breaks with \\

\author{Michael Manighalam} \email{mbmanigh@bu.edu}
 \affiliation{Department of Physics, Boston University}%Lines break automatically or can be forced with \\
\author{Mark Kon}%
 \email{mkon@bu.edu}
\affiliation{ 
Department of Mathematics and Statistics, Boston University%\\This line break forced with \textbackslash\textbackslash
}%

\date{\today}% It is always \today, today,
             %  but any date may be explicitly specified

\begin{abstract}
The discrete time quantum walk (DTQW) is a universal quantum computational model \cite{Lovett}. Significant relationships between discrete and corresponding continuous quantum systems have been studied since the work of Pauli and Feynman. This work continues the study of relationships between discrete quantum models and their ostensive continuum counterparts by developing a formal transition between discrete and continuous quantum systems through a formal framework for continuum limits of the DTQW. Under this framework, we prove two constructive theorems concerning which internal discrete transitions (``coins'') admit nontrivial continuum limits. We additionally prove that the continuous space limit of the continuous time limit of the DTQW can only yield massless states which obey the Dirac equation. Finally, we demonstrate that the continuous time limit of the DTQW can be identified with the canonical continuous time quantum walk (CTQW) when the coin is allowed to transition through the continuous limit process.
\end{abstract}

\maketitle

\section{\label{section:Intro}Introduction}

The discrete time quantum walk (DTQW) has been the subject of much attention since its applications to quantum computing were discovered
in the analysis of Hadamard Walks \cite{AmbainisDTQW}. The DTQW has since been used in a
variety of quantum computing algorithms, including the Oracular Search
\cite{Shenvi} and Element distinctness \cite{AmbainisElem} algorithms (for a full list, see Ref.~\onlinecite{QuantumAlgorithmZoo}).

As noted in Ref.~\onlinecite{Strauch}, a now well-studied limit of the DTQW was
introduced by Feynman and Hibbs in Ref.~\onlinecite{FeynHibbs} in constructing a
path integral formulation for the propagator of the Dirac equation. According to Feynman, a
particle zig-zags at the speed of light across a space-time lattice, flipping its
chirality from left to right with an infinitesimal probability at each time step
\cite{Strauch}. The Dirac equation results when the continuous space-time limit
is taken, with the mass of the particle determined by the flipping rate. More recent works have produced notions of discrete space-times (see Refs.~\onlinecite{CLDiscreteGeometries} and~\onlinecite{Nesterov}) and consequent questions regarding how they produce our apparent continuum.

Recently, Mlodinow and Brun in Ref.~\onlinecite{Mlodinowl} demonstrated how to constrain a 3D DTQW to obtain a resulting fully Lorenz invariant continuum limit. They showed that their symmetry requirement necessitates the inclusion of antimatter, and, in Ref.~\onlinecite{MlodinowlExperiment}, discuss experimental methods to distinguish between the DTQW and its continuum limiting Dirac equation as a description of fermion dynamics. These limits were also central to Refs.~\onlinecite{Knight} and \onlinecite{Blanchard}. Their continuum limits for DTQWs transformed discrete time evolution equations to partial differential
equations (PDEs), as the PDEs analyzed were much simpler than the discrete recursion relations of the DTQW.

In Ref.~\onlinecite{Strauch}, Strauch also used the continuum limit to
connect the DTQW and CTQW, and Refs.~\onlinecite{StrauchRelBig} and
\onlinecite{Bracken} demonstrate that the free particle Dirac evolution could be
obtained by taking continuum limits of the DTQW. Strauch also demonstrated, in
Ref.~\onlinecite{StrauchRelBig}, the DTQW's connections with zitterbewegung, 
which is an interference effect among free relativistic Dirac 
particles between their positive and negative energy parts
that produces a quivering motion \cite{Zitter}. Strauch shows that zitterbewegung
in the DTQW can be tuned based on the value of its coin rotation parameter, and
shows that the CTQW contains zitterbewegung-like oscillations (which Strauch
denotes as anomalous zitterbewegung) even though there is only one energy for the
CTQW \cite{StrauchRelBig}.

At this time, several forms of continuum limits have already
been rigorously developed, including  general space and time limits with coin
variations, in Refs.~\onlinecite{Molfetta} and \onlinecite{StrauchRelBig}, and the continuous time limit for a very particular choice of coin in Ref.~\onlinecite{Strauch}.

The purpose of this work is to formulate a general framework within which continuum limits of the DTQW can be taken, and to analyze the corresponding dynamics in the various limits. From our analysis, we have concluded that it is only possible to keep space
discrete while continuizing time for particular coins (section
\ref{section:CTLimit}); that taking time
and space limits simultaneously with a fixed coin is possible when steps in the
walk are allowed and yields a massless dirac equation (section \ref{section:noVariation}); that the ensuing time evolution derived from taking a continuous space limit of the continuous time limit of the DTQW is a massless dirac equation as well (section \ref{section:commutativity}); and that the solutions of
the continuous time limit of the DTQW can always be related to the solutions of the CTQW  for any choice of coin allowed to undergo the continuous time limit of the DTQW (section
\ref{section:DTQWtoCTQW}).

\subsection{DTQW Definition}

The one dimensional DTQW assumes a time dependent probability amplitude 
$\vv{\Psi}(x,t)=\begin{pmatrix}\psi_L(x,t)\\ \psi_R(x,t)\end{pmatrix}$ 
for a random walker's position and spin (assumed to point left or right).  Compared with the classical probabilistic random walk, this (i) involves an internal (left/right) spin degree of freedom and (ii) involves quantum probability amplitudes instead of classical random walk probabilities. The time dynamics are given as
\begin{equation}\label{eq:wave function}
\vv{\Psi}(x,t+\Delta t)=SC\vv{\Psi}(x,t),
\end{equation}
where the operations $S$ and $C$ (defined below) represent external and internal unitary operations, respectively, $S$ being an external translation operation and $C$ being an internal rebalancing of the two spin amplitudes $\psi_R$ and $\psi_L$.

For example, if the coin operation $C$ is implemented by the Hadamard matrix, then:

\begin{equation}
    C\vv{\Psi}(x,t)=\frac{1}{\sqrt{2}}\begin{pmatrix}1 & -1\\ 1 & 1\end{pmatrix}\vv{\Psi}(x,t).
\end{equation}
With $\Delta t$ and $\Delta x$ the time and space intervals for the quantum walk, the full change $SC$ acting in one time iteration $\Delta t$ is then:
\begin{equation}\label{eq:action}
    \begin{split}
        \vv{\Psi}(x,t+\Delta t)&=SC\vv{\Psi}(x,t)\\
        &=\frac{S}{\sqrt{2}}\begin{pmatrix}\psi_L(x,t)-\psi_R(x,t)\\\psi_L(x,t)+\psi_R(x,t)\end{pmatrix}\\
        &=\frac{1}{\sqrt{2}}\begin{pmatrix}\psi_L(x+\Delta x,t)-\psi_R(x+\Delta x,t)\\\psi_L(x-\Delta x,t)+\psi_R(x-\Delta x,t)\end{pmatrix}
    \end{split}
\end{equation}
The unitary time evolution is then:
\begin{equation}
     \vv{\Psi}(x,t)=(SC)^m\vv{\Psi}(x,0),
\end{equation}
letting $m=\frac{t}{\Delta t}$.
We also express this in discrete differential form for the purpose of forming subsequent continuum limits: 
\begin{equation}\label{eq:diffDTQW}
\Delta_t\vv{\Psi}(x,t)\equiv\frac{SC-\mathbb{I}}{\Delta t}\vv{\Psi}(x,t).
\end{equation}
We will also often represent the walk in Fourier space and define our discrete Fourier transform convention here. Let $\vv{\widetilde{\Psi}}(k,t)$ be the Fourier transform of $\vv{\Psi}(x,t)$ and $x=n\Delta x$ for $n\in\mathbb Z$.  We use the following conventions for the forward and inverse Fourier transforms, for Fourier variable $k\in[{-\pi \over \Delta x}, {\pi\over \Delta x}]$: 
\begin{align}
    \vv{\widetilde{\Psi}}(k,t)&=\sum_{n=-\infty}^{\infty}e^{-ikn\Delta x}\vv{\Psi}(n\Delta x,t)\equiv \mathcal F(\vv{\Psi})\\
    \vv{\Psi}(n\Delta x,t)&=\frac{\Delta x}{2\pi}\int^{\frac{\pi}{\Delta x}}_{-\frac{\pi}{\Delta x}}dke^{ikn\Delta x}\vv{\widetilde{\Psi}}(k,t)\equiv\mathcal{F}^{-1}(\vv{\widetilde \Psi}).
\end{align}
A standard procedure here will be to represent operators in Fourier space as follows: given an operator $O$ on a function space $Y$, its Fourier conjugate operator $\tilde O$ is defined by $\tilde O \tilde f(k)= \mathcal{F}(O(f(x)))$, with $f(x)\in Y$, so that $\tilde O$ is the Fourier representation of $O$. The operator we will be most commonly representing in Fourier space is the shift operator $S$, defined by $\widetilde{S}$:
\begin{equation}\label{eq:defn} 
\mathcal{F}(S\vv{\Psi}(x,t))=\widetilde{S}\vv{\widetilde{\Psi}}(x,t)=e^{ik\Delta x\sigma_z}\vv{\widetilde{\Psi}}(x,t), 
\end{equation}

where $\sigma_z$ is a Pauli matrix.

\section{Defining Continuum Limits}\label{sect:DefnContLim}
{\bf Skipping Steps.} Before formulating a universal definition of continuum limits for the quantum walk, we want to establish the important notion of so-called alternating limits, in which only steps of a certain parity (e.g. even or odd) are considered observed. We first provide an informal example demonstrating that trivial divergences occur in the $\Delta t \rightarrow 0$ limit arising from multiple parity-dependent limits in the discrete walk. Such limits were considered in Ref.~\onlinecite{Strauch}. 

Consider the DTQW with coin $C=ie^{i\theta \sigma_x}$, with $\sigma_x$ a standard Pauli matrix and $\theta\equiv\theta(\Delta t)$ a real number (modulo $2\pi$) depending on the time discretization parameter $\Delta t$. For the example we construct an informal continuous time limit, to be formalized in Definition \ref{def:CLimit}. Essentially we will take the $\Delta t\to 0$ limit in Equation (\ref{eq:diffDTQW}). The continuous time limit then amounts to identifying the limiting operator
$$\lim\limits_{\Delta t\to 0}\frac{\widetilde{S}(\Delta x)C(\Delta t)-\mathbb{I}}{\Delta t}=\lim\limits_{\Delta t\to 0}\frac{ie^{ik\Delta x\sigma_z}e^{i\theta(\Delta t)\sigma_x}-\mathbb{I}}{\Delta t},$$ 
assuming a fixed space of functions $\mathbb{X}$ on which it acts; here $\mathbb I$ is the identity. This is defined more carefully later in this section within {\bf Formal Definitions}.

For this analysis of a continuous time limit for the discrete space and time quantum walk, we will seek the most general scaling of walk parameters that admit nontrivial limits as $\Delta t\rightarrow 0$. In this case we will admit all scalings for the coin parameter of the form $\theta=\pi/2+\gamma \Delta t$, with $\gamma>0$, which were introduced by Strauch in Ref.~\onlinecite{Strauch}.  Thus from the operator standpoint we seek a limit of the form $\lim\limits_{\Delta t\to 0}\frac{ie^{ik\Delta x\sigma_z}i\sigma_x e^{i\gamma\Delta t\sigma_x}-\mathbb{I}}{\Delta t}$, which in fact does not exist generically.  We show here however, that if we consider only even parity steps (i.e. even numbers of steps, effectively considering only every other step), then non-trivial limits exist.  Thus we will be considering only iterations of the even parity operator $SCSC$ rather than the fundamental step $SC$, and we will identify a limit $\lim\limits_{\Delta t\to 0}\frac{\widetilde{S}(\Delta x)C(\Delta t)\widetilde{S}(\Delta x)C(\Delta t)-\mathbb{I}}{\Delta t}$, which structurally is:
\begin{align*}
&\lim\limits_{\Delta t\to 0}\frac{\widetilde{S}(\Delta x)C(\Delta t)\widetilde{S}(\Delta x)C(\Delta t)-\mathbb{I}}{\Delta t}=\lim\limits_{\Delta t\to 0}\frac{-e^{ik\Delta x\sigma_z}i\sigma_x e^{i\gamma\Delta t\sigma_x}e^{ik\Delta x\sigma_z}i\sigma_x e^{i\gamma\Delta t\sigma_x}-\mathbb{I}}{\Delta t}\\
=&\lim\limits_{\Delta t\to 0}\frac{ (\mathbb{I}+i\gamma\Delta t (\sigma_x\cos{2k\Delta x}+\sigma_y\sin{2k\Delta x})+O(\Delta t^2)) (\mathbb{I}+i\gamma\Delta t\sigma_x+O(\Delta t^2)-\mathbb{I})}{\Delta t}\\
=&i\gamma(\sigma_x(\cos{2k\Delta x}+1)-\sigma_y\sin{2k\Delta x}).
\end{align*}

As might be expected, it will be clear below that replacing the above even power $(SC)^n$ with $n=2$ by $n=3$, the above limiting process will no longer exist; existence of the limit will hold only for even powers $n$. In general, restricting to fixed even step sizes $n$ will lead to continuous limiting processes as above (with scaling of the coin based on $\Delta t$), while non-even step sizes will never admit such limits (see Theorem \ref{thm:ctdsUnitaries}, proved in Appendix \ref{app:ctdsUnitaries}).

{\bf Formal Definitions.  }  Definitions of our operator limits require common spaces for their domains.  We will redefine all operators on such a common space, given as
$$ \mathbb{X}=\{\vv{\Psi}(x,t):\vv{\Psi}(\cdot,t)\in L^2(\mathbb{R})\otimes L^2(\Sigma)\text{ for all } t\ge 0\text{ and } \vv{\Psi}(x,\cdot)\in C^1(\mathbb{R}\to\mathbb{C}^2)\},$$ where $\Sigma$ is the space spanned by $\ket{L}=\begin{pmatrix}1\\0\end{pmatrix}$ and $\ket{R}=\begin{pmatrix}0\\1\end{pmatrix}$. Note that the effective domain space of the above tensor product space is $\R\times\Sigma$, with $\Sigma={L,R}$. Thus $\vv\Psi (x,t)$ is assumed once continuously differentiable in $t$, with two components in $L^2$ (i.e. square integrable functions in $x\in \R$ for fixed $t$).

We will consider general quantum walks that have $\Delta t\rightarrow 0$ limits when step numbers $n=km$ are restricted to whole multiples of an integer $n$, i.e.\  generalizing the above parity restriction for step numbers $(n=2)$ to accommodate more general step number restrictions.
Thus let $\vv{\Psi}(x,t)\in \mathbb{X}$, $n$ be the number of skipped steps, $\partial_t$ be the time derivative operator, and define the discrete derivative as $\Delta_t\vv{\Psi}(x,t)=\frac{\vv{\Psi}(x,t+n\Delta t)-\vv{\Psi}(x,t)}{n\Delta t}$. If $\vv{\Psi}(x,t)\in \mathbb{X}$ is a  wave function, then the DTQW time evolution equation is
\begin{equation}\label{eq:DTQWevo}
     \vv{\Psi}(x,t+n\Delta t)=(S(\Delta x)C(\Delta t))^n\vv{\Psi}(x,t).
\end{equation} We denote the level of discretization of our space and time operations by $\eta=(\Delta x, \Delta t)$.
We will consider a discrete space and time quantum walk governed by  
\begin{equation}\label{eq:Heps}
    i\Delta_t\vv{\Psi}=H_\eta\vv{\Psi}\equiv i\frac{\big(SC\big)^n-\mathbb{I}}{n\Delta t}\vv{\Psi}(x,t)
\end{equation}
on $\mathbb X$, with $H_\eta$ the above family of operators parametrized by $\eta=(\Delta x,\Delta t)$. The continuous time limit of the walk in Equation (\ref{eq:Heps}) exists if the right hand side of the equation has a limit (for $\vv{\Psi}\in\mathbb{X}$) as $\eta\to(0,0)$ along a given prescribed path, for which both the DTQW functions and continuum limit of the DTQW functions are in $\mathbb X$. Continuum space limits in the absence of any change
in $\Delta t$ will not be considered here because $S\to\mathbb{I}$ as $\Delta x \to 0$, so the
walk reduces simply to a coin acting on the spin portion of the wave function at each time step.
With no traversal of the lattice there results a trivial walk. Formally, we state the definition of continuous time limit and continuous space-time limit as:
%Continuum Limit Definition; note that we can have theorems related to existence of limits along three particular types of paths as $\epsilon \righarrow 0$
\begin{definition}\label{def:CLimit}\label{def:Exist}
Let the operators $H_\eta\equiv H_{\Delta x,\Delta t}$ and $H_{\Delta t}$ act on functions $\vv{\Psi}(x,t)\in \mathbb{X}$. Then we have the following definitions:
\begin{itemize}
    \item The \textbf{continuous time limit of the DTQW governed by coin $C$ skipping $n$ steps} is the time evolution equation $i\partial_t\vv{\Psi}(x,t)=H_{\Delta t}\vv{\Psi}(x,t)$ where $H\Delta t$ is defined (when the limit exists) by $H_{\Delta t}\vv{\Psi}(x,t)=\lim\limits_{\eta\to({\Delta x, 0})}H_\eta\vv{\Psi}(x,t)$, with the limit taken in the space $\mathbb X$.
    \item The \textbf{continuous space-time limit of the DTQW governed by coin $C$ skipping $n$ steps} is the time evolution equation  $i\partial_t\vv{\Psi}(x,t)=H_{\Delta x,\Delta t}\vv{\Psi}(x,t)$, where $H_{\Delta x,\Delta t}$ is is defined (when the limit exists) by $H_{\Delta x,\Delta t}\vv{\Psi}(x,t)=\lim\limits_{\eta\to({0, 0})}H_\eta\vv{\Psi}(x,t)$, (where in the limit $\Delta x=v\Delta t$ for some $v>0$).
\end{itemize}
We call the operators $H_{\Delta t}$ and $H_{\Delta x,\Delta t}$ the generators of time evolution, or Hamiltonians, in their respective continuum limits.  Note that the second limit above may depend on the ratio $\nu=\frac{\Delta x} {\Delta t}$, and can also be generalized to allow any manner of approach of $\eta\rightarrow (0,0)$.  
\end{definition}

Our goal is to explore the most general possibilities for these two cases. We remark that our inclusion of $n$ expands the number of continuum limits that exist; in particular this possibility was not considered in Ref.~\onlinecite{Molfetta}

Additionally, we need the following definition to allow parameterized coin variations:
%Coin Variation Definition%
\begin{definition}\label{def:Variation} Consider a continuous space-time limit where $\Delta x$ and $\Delta t$ have the same scaling, so $\Delta x=v\Delta t=v\epsilon$ for some non-zero $v\in\mathbb{R}$. A coin \textbf{varies} in this continuum limit if the coin depends on $\epsilon=\Delta t$.
\end{definition}

%Extensive work has been done on continuum limits, as listed in section \ref{section:Intro}.
%We will be using these definitions to add to the body of knowledge of continuum limits of the DTQW. For the continuous time limit, Strauch connected the continuous time limit of the DTQW to the continuous time quantum walk, or CTQW (as defined in Ref.~\onlinecite{FarhiCTQW}). Strauch used a DTQW defined by a specific coin to obtain this result, and we add to this idea in \textbf{section \ref{section:CTLimit}} by determining which quantum walks in general can undergo such a limit. Additionally, we address in \textbf{section \ref{section:commutativity}} the differences between time and then space limits vs simultaneous space-time limits of the DTQW.

%For the continuous space-time limit, Ref.~\onlinecite{Molfetta} found a very general form of the continuous space-time limit when skipping steps is not allowed, and they concluded that there must be coin variation. We aim to add to this body of work in \textbf{section \ref{section:noVariation}} by determining in more generality which DTQWs can undergo the continuous space-time transformation with non-varying coin when skipped steps are allowed, and show that the ensuing hamiltonian is a massless dirac hamiltonian. A brief discussion of Ref.~\onlinecite{Molfetta}'s relevant results can be seen in Appendix \ref{App:CSTn=1}.

\section{General Conditions for Continuum Limits}\label{section:rootsofunity}
The discussion here is based on terminology and results in Ref.~\onlinecite{Molfetta}. We will study an important aspect of coins that change in the process of continuous time and space-time limits; this will help to interpret the theorems in Sections \ref{section:CTLimit} and \ref{section:noVariation}. We will follow the DTQW wave function through $n$ time steps of length $\Delta t$. Note here that all limits in this section will be in the topology of the space $\mathbb{X}$.  We begin with the basic equation
\begin{equation}\label{eq:1}
    \vv{\Psi}(x,t+n\Delta t)=(S(\Delta x)C(\Delta t))^n\vv{\Psi}(x,t),
\end{equation}
with $S=S(\Delta x)$ and $C=C(\Delta t)$ both dependent on the increment $\eta = (\Delta x, \Delta t)$.
If a continuous space-time limit is taken with $(\Delta t$, $\Delta x)\to(0,0)$, then because $S\to\mathbb{I}$ when $\Delta x\to 0$, we must have 
$$\lim_{\Delta t,\Delta x\to 0}(C(\Delta t))^n=\mathbb{I},$$
as otherwise the limit could not exist. In particular, unless $C(\Delta t)$ is constantly the identity, it must vary (as in Definition \ref{def:Variation}) in the continuous time limit.

If only a continuous time limit is taken (i.e. $\Delta t\to 0$), then (by continuity of the left side in $t$) for the left and right sides of Equation (\ref{eq:1}) to be equal, we must have $$\lim_{\Delta t\to 0}(S(\Delta x)C(\Delta t))^n=\mathbb{I}$$ where we include $\Delta t$ dependence in $C$ for generality. Note that the constraint in the continuous time limit involves both the coin and the shift operator, not just the coin as in the continuous space-time limit. Now for the following definition:
\begin{definition}\label{def:rootOfUnity}
Consider a matrix $A(t)$ which depends on some continuous parameter $t$. $A(t)$ \textbf{homotopically approaches a root of unity} if $A(t)$ depends continuously on $t$ and there exists some non-zero integer $m$ and some real number $t'$ such that $\lim\limits_{t\to t'}A(t)^m=\mathbb{I}$.
\end{definition}
By the previous definition and the above analysis of Equation (\ref{eq:1}), we have the following theorem:
%(MK) Stopped here on 4/28; look at rationale for replacing S by I below
\begin{theorem}\label{theorem:rootofunity}
A coin for which a \textbf{continuous space and time limit} exists must homotopically approach a root of unity. The product of the shift and coin operator for which a \textbf{continuous time limit} exists must homotopically approach a root of unity as well.
\end{theorem}
\begin{proof}
Recall from definition \ref{def:CLimit} that we define the space-time limit $H_{\Delta x,\Delta t}$ with $\Delta x=v\Delta t=v\epsilon$ as:
\begin{equation}
   H_{\Delta x,\Delta t}\vv{\Psi}(x,t)\equiv i\lim_{\epsilon\to 0}\frac{(S(\epsilon)C)^n-\mathbb{I}}{n\epsilon}\vv{\Psi}(x,t)
\end{equation}
for $\vv{\Psi}(x,t)\in\mathbb{X}$.
Because $\lim\limits_{\epsilon\to0}S=\mathbb{I}$, we have the following:
\begin{equation}
    i\lim_{\epsilon\to0}\frac{C^n-\mathbb{I}}{n\epsilon}\vv{\Psi}(x,t)=H_{\Delta x,\Delta t}\vv{\Psi}(x,t)
\end{equation}
Now we see that for the left hand side to equal the right hand side, $C$ must be of the form $C^n=\mathbb{I}-in\epsilon H_{\Delta x,\Delta t}+O(\epsilon^2)$. Thus, by definition \ref{def:rootofunity}, $C$ must homotopically approach a root of unity in the continuous space-time limit. The proof for the continuous time limit is similar, except now $S$ does not converge to identity, so instead $SC$ must be of the form $(SC)^n=\mathbb{I}-in\epsilon H+O(\Delta \epsilon^2)$, thereby satisfying definition \label{def:rootofunity} once again.
\end{proof}
From this analysis, we have obtained a general property of coins which undergo continuum limit transformations, and we will refer to this property in the future.

\section{General Continuous Time Limit}\label{section:CTLimit}
In this section, we will find the set of DTQWs for which a continuous time limit exists, as according to definition \ref{def:Exist}. We will then analyze the properties of the resulting time evolution in the continuous time limit.

We consider a general unitary coin, which can be written the following way:
\begin{equation}\label{eq:rotations}
    \begin{split}
        C&=e^{i\delta}R_z(\psi)R_y(\theta)R_z(\phi)=e^{i\delta}e^{ -i\psi\sigma_z/2}e^{ -i\theta\sigma_y/2} e^{ -i\phi\sigma_z/2}\\
        &=e^{i\delta}\begin{pmatrix}\cos\frac{\theta}{2}\exp-i\frac{\phi+\psi}{2}&-\sin\frac{\theta}{2}\exp i\frac{\phi-\psi}{2}\\\sin\frac{\theta}{2}\exp i\frac{-\phi+\psi}{2}&\cos\frac{\theta}{2}\exp i\frac{\phi+\psi}{2}\end{pmatrix}
    \end{split}
\end{equation}
We wish to know for which $2\times2$ matrices, as parameterized by Equation (\ref{eq:rotations}), does the continuum limit exist, according to definition \ref{def:Exist}.
Before introducing a theorem answering such a problem, a few notes. First, a constraint on $\delta$ is necessary to satisfy the finiteness condition of the existence of the limit in definition \ref{def:Exist}. The value of $\delta$ is arbitrary as it amounts to an overall energy shift in the Hamiltonian, which does not matter. This point is explained further in the proof of lemma \ref{lma:angles}. Second, assuming that the elements of $C$ cannot depend on the elements of $S$ in any way, observe that the limit in definition \ref{def:Exist} cannot not be finite unless $C$ depended on $\Delta t$ in some way. Thus, we must assume that the coin varies in the continuum limit, where variation is defined in \ref{def:Variation}. A proof of the following theorem is presented in Appendix \ref{app:ctdsUnitaries}.

\begin{restatable}{theorem}{ctdsU}\label{thm:ctdsUnitaries}
Let $C(\delta,\psi,\theta,\phi)$ be a $2\times2$ unitary matrix as defined in Equation (\ref{eq:rotations}) such that the set of angles {$\psi$, $\theta$, $\phi$} parameterizing $C$ depends on $\Delta t$ the following way: $\phi=\phi_0+\phi_1\Delta t+O(\Delta t^2)$, $\psi=\psi_0+\psi_1\Delta t+O(\Delta t^2)$, and $\theta=\theta_0+\theta_1\Delta t+O(\Delta t^2)$, where $\phi_0,\psi_0,\theta_0,\phi_1,\psi_1,\theta_1\in\mathbb{R}$ do not have any space or time dependence. The continuous time limit will exist, as defined in \ref{def:Exist}, for such a class of coins  if and only if, $\theta_0=p\pi$, $\delta=-\frac{p\pi}{2}$ (for odd integer $p$), and $n$ is even. The Hamiltonian obtained in such a limit is the following, where $S$ is the shift operator defined in equation
\ref{eq:defn}:
\begin{equation}
    H=-\frac{\theta_1}{4}(R_z(-2\phi_0)+S^2 R_z(2\psi_0))\sigma_y.
\end{equation}
\end{restatable}

When using the conditions for the coin from Theorem \ref{thm:ctdsUnitaries}, we see that the class of coins which can undergo continuous time limits, as determined in the theorem, are the following (parameterized by $\psi_0$, $\theta_1$, and $\phi_0$):

\begin{equation}\label{eq:CTcoin}
    C=e^{-i(\psi_0)\sigma_z/2}\sigma_y e^{-i\theta_1\Delta t\sigma_y/2} e^{-i(\phi_0)\sigma_z/2}
\end{equation}
An interesting observation is that the final hamiltonian after the continuous time limit is taken does not depend on any parameters that are coefficients of terms $O(\Delta t)$ except $\theta_1$ (so they really do not need to be included in Equation (\ref{eq:CTcoin}). $\theta_1$ can be interpreted as a driving factor for the final hamiltonian's time evolution. Its value completely determines how much mixing between the left and right states there will be due to the time evolution operator $SC$, as when $\theta_1=0$ all operators in the coin commute with the shift operator and no mixing occurs, which corresponds to the wave function coming back to itself every other step in the DTQW. Another interesting point is the reason for the skipping of steps (i.e. $n$ must be even). An explicit proof of why this must be can be seen in Appendix \ref{app:ctdsUnitaries}, but for a more intuitive explanation, consider the following. For the continuous time limit of the DTQW to exist, the following Fourier space hamiltonian must be finite:
\begin{equation}\label{eq:intuitiveCT}
    \tilde{H}(k)=\lim_{\Delta t\to 0}\frac{(e^{ik\Delta x\sigma_z}C(\Delta t))^n-\mathbb{I}}{\Delta t}
\end{equation}
The operator $e^{ik\Delta x\sigma_z}C$ does not homotope to identity as $\Delta t\to 0$ for any $C$, so no continuous time limit can exist for $n=1$. However, for the coin in Equation (\ref{eq:CTcoin}), the operator $e^{ik\Delta x\sigma_z}Ce^{ik\Delta x\sigma_z}C$ does homotope to the identity because of the following property of the coins in Equation (\ref{eq:CTcoin}): $Ce^{ik\Delta x\sigma_z}C=e^{-ik\Delta x\sigma_z}+O(\Delta t)$ i.e. the coins derived in Theorem \ref{thm:ctdsUnitaries} invert the shift operator up to $O(\Delta t)$. This is what causes the $O(\Delta t^0)$ term in $SCSC$ to become exactly identity. After the identities cancel in Equation (\ref{eq:intuitiveCT}), the $O(\Delta t)$ term is all that remains, which is the hamiltonian in \ref{thm:ctdsUnitaries}.

It should be noted that Ref.~\onlinecite{Strauch} derives a special case of the Hamiltonian in Theorem \ref{thm:ctdsUnitaries} using the coin $C=e^{-i\theta\sigma_x}$ and has $\theta=\frac{\pi}{2}-\gamma\Delta t$. The values of the angles parameterizing the general unitary coin in Theorem \ref{thm:ctdsUnitaries} for the particular choice of coin in Ref.~\onlinecite{Strauch} and $\Delta t$ dependence of angles are the following (where $\gamma$ is the jumping rate from vertex to vertex):
\begin{equation}\label{eq:StrauchWalk}
    \begin{split}
        &\psi_0=-\frac{\pi}{2},~\phi_0=\frac{\pi}{2},~\theta_0=\pi-2\gamma\delta t\\
        &\psi_1=0,~\phi_1=0,~\theta_1=4\gamma\\
        &\delta=0
    \end{split}
\end{equation}
The repercussions of not setting $\delta=p\pi$ for odd integer p leads to a final $H$ having constant infinite energy contributions, which are then ignored, as energy differences are the physical quantities.

%One can check that the family of coins for which the continuous time limit exists, as derived in theorem \ref{thm:ctdsUnitaries}, when multiplied by the corresponding shift operator and raised to an even integer $n$, homotopically approach a root of unity as $\Delta t\to 0$ (i.e. $\Delta t$ is the continuous parameter and the product of $(SC)^n$ limits to $\mathbb{I}$ as $\Delta t$ approaches $t'=0$). This follows from our analysis of the continuous time limit in section \ref{section:rootsofunity}.

Another interesting property of the coins derived in \ref{thm:ctdsUnitaries} is that the coins themselves homotopically approach a root of unity as well, so $\lim\limits_{\Delta t\to 0}C^n=\mathbb{I}$, as one can check. This was not implied by our analysis in section \ref{section:rootsofunity} for the continuous time limit case, and a full derivation of all the possible coins which can undergo continuous time limits was needed to obtain this property. Also, the resulting hamiltonian in Theorem \ref{thm:ctdsUnitaries} will be used in section \ref{section:commutativity} to determine how the continuous time limit followed by the continuous space limit compares to the simultaneous continuous space-time limit.

Now we analyze wave functions which undergo time evolution dictated by the hamiltonian obtained in theorem \ref{thm:ctdsUnitaries}. A proof of the following corollary is found in Appendix~\ref{app:corrctdsWavefunction}.

\begin{restatable}{corollary}{ctdsWaveFunctions}\label{corr:ctdsWavefunction}
Let $\vv{\Psi}(x,t)$ be a solution to the time evolution equation with the hamiltonian from theorem \ref{thm:ctdsUnitaries}, $i\partial_t\vv{\Psi}(x,t)=H\vv{\Psi}(x,t)=-\frac{\theta_1}{4}(R_z(-2\phi_0)+S^2 R_z(2\psi_0))\sigma_y\vv{\Psi}(x,t)$. Also, let $\vv{\Psi}(x,0)=\begin{pmatrix}\Psi_L(x,0)\\\Psi_R(x,0)\end{pmatrix}$ be the initial condition for $\vv{\Psi}(x,t)$. Then the following is the analytical form of the time evolution for $\vv{\Psi}(x,t)$ for all $t$ in terms of its initial state, where $x=m\Delta x$ for $m\in\mathbb{Z}$ and $J_m(t)$ is the $m^{\text{th}}$ order Bessel function of the first kind, $\alpha=\frac{\phi_0+\psi_0}{2}$, and $\beta=\frac{\phi_0-\psi_0}{2}$:
\begin{align*}
    &\vv{\Psi}(m\Delta x,t)=\begin{pmatrix}\Psi_L(m\Delta x,t)\\\Psi_R(m\Delta x,t)\end{pmatrix}=\\
    &\frac{1}{2}\sum_{n=-\infty}^\infty i^{m-n}e^{i\alpha(m-n)}J_{m-n}(\frac{\theta_1t}{2})\begin{pmatrix}(1+(-1)^{m-n})\Psi_L(n\Delta x,0)+ie^{i\beta}(1-(-1)^{m-n})\Psi_R((n+1)\Delta x,0) \\-ie^{-i\beta}(1-(-1)^{m-n})\Psi_L((n-1)\Delta x,0)+(1+(-1)^{m-n})\Psi_R(n\Delta x,0) \end{pmatrix}
\end{align*}
\end{restatable}
This solution does reduce to that found in Ref.~\cite{Strauch} when the corresponding parameters in equation \ref{eq:StrauchWalk} are used, except for some minor sign differences stemming from their shift operator being defined as the inverse of our shift operator. We see that the locations for which $\vv{\Psi}_L(x,0)$ is nonzero will contribute to $\vv{\Psi}_L(m\Delta x,t)$ if they are an even number of steps away from $m$, and the nonzero locations of $\vv{\Psi}_R(x,0)$ will contribute to $\vv{\Psi}_L(m\Delta x,t)$ if they are an odd number of steps away from $m$, and the opposite scenario is true for $\vv{\Psi}_R(m\Delta x,t)$. For a full description of the effects $\alpha$ and $\beta$ have on the probability distribution, see section \ref{section:DTQWtoCTQW}.

%%%%%%%%%%%%%%%%%%%%%%%%%%%%%%%%%%%%%%%%%%%%%%%%%%%%%%%%%%%%%%%%%%%%%%%%%%%%%%%%%%%%%%%%%%%%%%%%%%%%%%%%%%%%%%%%%%%%%%%%%%%%%%

\section{Continuous Space-Time Limit with No Coin Variation}\label{section:noVariation}

In this section we will show for which DTQWs the continuous space-time limit exists and what the ensuing time evolution is if there is no coin variation involved, as defined in definition \ref{def:Variation}. We present this theorem to show that it is possible to obtain a continuous space-time limit of a DTQW with non-varying coin, and to demonstrate properties that coins must have to undergo this type of limit. A proof of the following theorem is presented in Appendix \ref{app:FixedCoinTheorem}.

\begin{restatable}{theorem}{noVariation}\label{thm:fixedCoinTheorem}
Let $\vv{\Psi}(x,t)$ be a 2 component wave function undergoing the DTQW, as defined in Equation (\ref{eq:wave function}). Also, let $|\hat{n}|=\sqrt{n_x^2+n_y^2+n_z^2}=1$, $l=0,1,2,...$, $m=1,2,...$, and $v\Delta t=\Delta x$, where $\Delta t$ and $\Delta x$ are the time step and lattice spacings of the DTQW for $\vv{\Psi}(x,t)$, respectively. The continuous space-time limit will exist for $\vv{\Psi}(x,t)$ if and only if the DTQW skips every $m$ steps and the coin operator dictating its DTQW is of the form
\begin{equation}\label{eq:ctcsCoin}
    C=\exp \frac{i\pi l}{m}\exp \frac{-i\pi l}{m}\hat{n}\cdot\vv{\sigma}.
\end{equation}
The ensuing Hamiltonian for such a walk will be the following massless Dirac hamiltonian: 
\begin{equation}
    H=-vn_z\hat{n}\cdot\vv{\sigma}\frac{\partial}{\partial x}.
\end{equation}
\end{restatable}

We see that the massless Dirac hamiltonian is the limiting hamiltonian of this continuum limit. In the continuous space-time limit the mass term is generated by the coin's variation with time step in the continuum limit, as can be seen in Appendix \ref{App:CSTn=1}. Because our coin does not vary in the continuum limit, the ensuing continuous space-time hamiltonian will have no mass.

Another interesting note is that the above theorem states that a coin with no variation will have a continuum limit if it is itself a root of unity. This makes sense in light of Theorem \ref{theorem:rootofunity}, as the continuous parameter in our coin is no longer there, so the coin itself must be a root of unity. Also, this theorem might seem at odds with the discussion at the start of Ref.~\onlinecite{Molfetta} (which we repeat in section \ref{section:rootsofunity}) but they did not consider skipping steps in the walk when taking the continuum limit, which is how we were able to get a limit for this walk even when our coin did not vary at all in the continuum limit.

%%%%%%%%%%%%%%%%%%%%%%%%%%%%%%%%%%%%%%%%%%%%%%%%%%%%%%%%%%%%%%%%%%%%%%%%%%%%%%%%%%%%%%%%%%%%%%%%%%%%%%%%%%%%%%%%%%%%%%%%%%%%%%

\section{Simultaneous Continuous Space-Time Limit vs Continuous Time Followed by Continuous Space Limit}\label{section:commutativity}
 In this section we state a theorem on the existence of non-trivial continuous space limits of the continuous time limit of the DTQW. We begin with the theorem (proof in Appendix \ref{App:CTCS}):

\begin{restatable}{theorem}{ctcs}\label{thm:ctcs}
Let $\phi_0$ and $\psi_0$ be unable to vary in the continuous space limit (i.e. $\phi_0$, $\psi_0$ cannot depend on $\Delta x$). Then the only time evolution equation which is not infinite and contains spatial derivative(s) for the continuous space limit ($\Delta x\to0$) of the continuous time limit of the DTQW is a massless dirac equation.
\end{restatable}
The reason why we do not allow for $\phi_0$ and $\psi_0$ to not depend on $\Delta x$ is given by the following conjecture:
\begin{conjecture}\label{conjecture:phipsi}
There is no dependence $\phi_0$ and$/$or $\psi_0$ can have on $\Delta x$ that would allow for spatial derivative(s) in the continuum limit
\end{conjecture}
The reason why we are searching for time evolution equations with spatial derivatives is because without them, no spatial translation will occur for our wave function in the continuous space limit, thus resulting in a trivial stationary walk. An interesting note of Theorem \ref{thm:ctcs} is that a different time evolution equation occurs when a \textbf{simultaneous} continuous space-time limit is taken. As can be seen in Appendix \ref{App:CSTn=1}, when a simultaneous space-time continuum limit is taken, a massive Dirac equation results.

\section{General DTQW Relationship to CTQW}\label{section:DTQWtoCTQW}
In the following section, we will be building on a result of Strauch's from Ref.~\onlinecite{Strauch}, in which a connection was found between the DTQW and CTQW by taking a continuous time limit of the DTQW to relate it to the CTQW. Strauch used a specific coin $e^{-i\theta\sigma_x}$, and let $\theta=\frac{\pi}{2}-\gamma\Delta t$ when the continuous time limit was taken. Now that a general parameterization of all the possible coins which can undergo a continuous time quantum walk has been obtained from Theorem \ref{thm:ctdsUnitaries}, we have the opportunity to investigate if, for a general coin, whether or not a relationship between the CTQW and DTQW exists. We begin by reviewing Strauch's specific results in Ref.~\onlinecite{Strauch}.

%%%%%%%%%%%%%%%%%%%%%%%%%%%%%%%%%%%%%%%%%%%%%%%%%%%%%%%%%%%%%%%%%%%%%%%%%%%%%%%%%%%%%%%%%%%%%%%%%%%%%%%%%%%%%%%%%%%%%%%%%%%%%%

\subsection{Review of Strauch}\label{section:StrauchReview}
In this section, we will be reviewing the connection between the DTQW and CTQW made by Strauch in Ref.~\onlinecite{Strauch}. To start, consider a DTQW dictated by the shift operator (in Fourier space) $\widetilde{S}=e^{ik\Delta_x\sigma_z}$ and coin operator $C=e^{-i\theta\sigma_x}$ such that the time evolution of a Fourier space wave function $\vv{\widetilde{\Psi}}(k,t)$ is given by $\vv{\widetilde{\Psi}}(k,t+\Delta t)=\widetilde{S}C\vv{\widetilde{\Psi}}(k,t)$. When a continous time limit ($\Delta t\to 0$) is taken on $\vv{\Psi}(x,t)$, letting $\theta=\frac{\pi}{2}-\gamma\Delta t$ and skipping every other step, the following time evolution is recovered for $\vv{\Psi}(x,t)$:
\begin{equation}
    i\partial_t\vv{\Psi}(x,t)=-\gamma(\mathbb{I}+S^2)\sigma_x\vv{\Psi}(x,t)
\end{equation}
Now for the profound relation discovered by Strauch. If we define two wave functions $\vv{\Psi}_+(x,t)$ and $\vv{\Psi}_-(x,t)$ such that $\vv{\Psi}_{\pm}(x,t)\equiv\frac{e^{\mp2i\gamma t}}{2}(\mathbb{I}\pm S\sigma_x)\vv{\Psi}(x,t)$, then it can be shown that $\vv{\Psi}(x,t)=e^{2i\gamma t}\vv{\Psi}_+(x,t)+e^{-2i\gamma t}\vv{\Psi}_-(x,t)$ and $i\partial_t\vv{\Psi}_\pm(x,t)=\mp\gamma\big[\vv{\Psi}_\pm(x+\Delta x,t)+\vv{\Psi}_\pm(x-\Delta x,t)-2\vv{\Psi}_\pm(x,t)\big]$ (which is the CTQW time evolution equation). In other words, Strauch found that the continuous time limit of his DTQW can be written as a superposition of two copies of the CTQW. This relation helped clarify the then longstanding mystery about the exact relationship between the two ways of quantizing the quantum walk, the DTQW and CTQW. Next we show that this relationship holds for a general coin, and we will use the relation to see how the DTQW coin parameters effect the final probability distribution in the discussion following the theorem \ref{thm:connection}.

%%%%%%%%%%%%%%%%%%%%%%%%%%%%%%%%%%%%%%%%%%%%%%%%%%%%%%%%%%%%%%%%%%%%%%%%%%%%%%%%%%%%%%%%%%%%%%%%%%%%%%%%%%%%%%%%%%%%%%%%%%%%%%

\subsection{General Coin CTQW-DTQW relation}\label{section:generalCoinRelation}
We summarize our findings in the following theorem, the proof of which is in Appendix \ref{app:CTDTrelation}:
\begin{restatable}{theorem}{ctdtqwRelation}\label{thm:connection}
Let $\vv{\Psi}(x,t)$ be the following 2 component wave function resulting from the continuous time limit of the DTQW with a general coin as found in Theorem \ref{thm:ctdsUnitaries}: $$ i\partial_t\vv{\Psi}(x,t)=-\frac{\theta_1}{4}(R_z(-2\phi_0)+S^2 R_z(2\psi_0))\sigma_y\vv{\Psi}(x,t)$$ where $\theta_1$, $\phi_0$, and $\psi_0$ are real numbers which cannot depend on $x$ or $t$. Also let $\vv{\Psi}_\pm(x,t)$ be wave functions which satisfy the following CTQW time evolution equations: $$i\partial_t\vv{\Psi}_\pm(x,t)=\mp\frac{\theta_1}{4}\big[\vv{\Psi}_\pm(x+\Delta x,t)+\vv{\Psi}_\pm(x-\Delta x,t)-2\vv{\Psi}_\pm(x,t)].$$ Then $\vv{\Psi}(x,t)$ can be written as a superposition of $\vv{\Psi}_+(x,t)$ and $\vv{\Psi}_-(x,t)$ in the following way (where $\alpha=\frac{\phi_0+\psi_0}{2}$):
\begin{equation}\label{eq:ctqw-dtqwRelation}
    \vv{\Psi}(x,t)=e^{i\alpha\frac{x}{\Delta x}}(e^{-\frac{i\theta_1 t}{2}}\vv{\Psi}_+(x,t)+e^{\frac{i\theta_1 t}{2}}\vv{\Psi}_-(x,t)).
\end{equation}
\end{restatable}
Now that we have a general relationship between the continuous time limit of the DTQW and the CTQW, we can analyze exactly how the coin parameters $\alpha$ and $\beta$ effect the probability distribution of the continuous time limit. First of all, $\vv{\Psi}_\pm(x,t)$ are fixed momentum traveling wave states with time evolution which does not depend on $\alpha$ or $\beta$ because $\vv{\Psi}_\pm(x,t)$ satisfy the CTQW (which does not depend on $\alpha$ or $\beta$), so the time evolution of just these wave functions would be to just spread their initial distribution across the sites. Now we write equation \ref{eq:ctqw-dtqwRelation} more suggestively:
\begin{equation}
    \vv{\Psi}(x,t)=e^{i(\alpha\frac{x}{\Delta x}-\frac{\theta_1 t}{2})}\vv{\Psi}_+(x,t)+e^{i(\alpha\frac{x}{\Delta x}+\frac{\theta_1 t}{2})}\vv{\Psi}_-(x,t).    
\end{equation}
$e^{i(\alpha\frac{x}{\Delta x}-\frac{\theta_1 t}{2})}$ has the effect of boosting $\vv{\Psi}_+(x,t)$ to a frame traveling right (if $\alpha>0$) at speed $|\frac{\alpha\theta_1}{2}|$ or left (if $\alpha<0$), and $e^{i(\alpha\frac{x}{\Delta x}+\frac{\theta_1 t}{2})}$ boosts $\vv{\Psi}_-(x,t)$ to a frame moving at speed $|\frac{\alpha\theta_1}{2}|$ in the opposite direction as $\vv{\Psi}_+(x,t)$. The last effect these parameters will have will be on the initial condition of $\vv{\Psi}_\pm(x,t)$. The operator which projects $\vv{\Psi}(x,0)$ onto $\vv{\Psi}_\pm(x,0)$ is $P_\pm=e^{-i\alpha\frac{x}{\Delta x}}(\frac{1}{2}\mp S e^{i\beta z}y)$, so $P_\pm\vv{\Psi}(x,0)=\vv{\Psi}_\pm(x,0)$. We see the only effect $\beta$ has is on the initial conditions, while $\alpha$ effects both the initial condition and the frames $\vv{\Psi}_+(x,t)$ and $\vv{\Psi}_-(x,t)$ are boosted to.

\section{Conclusion and Open Questions}
\subsection{Conclusions}
Given our definitions of continuum limit from section \ref{sect:DefnContLim}, we have concluded by theorem \ref{thm:ctdsUnitaries} that keeping space discrete while continuizing time is only possible for particular coins, which must be of the form $C=e^{-i(\psi_0)\sigma_z/2}\sigma_y e^{-i\theta_1\Delta t\sigma_y/2} e^{-i(\phi_0)\sigma_z/2}$, granted the limit is taken 2 steps at a time. We have also concluded, from theorem \ref{thm:connection} in section \ref{section:generalCoinRelation}, that the continuous time limit of the DTQW can be always be related to the CTQW if the coin is of the form $\exp-\frac{i\theta}{2}(\sigma_y\cos\psi_0-\sigma_x\sin\psi_0)$, where $\theta\to p\pi+\theta_1\Delta t$ in the continuum limit for some odd integer $p$ and $\theta_1\in\mathbb{R}$. The implications of these two theorems are that there exists unitary matrices used as coins in the DTQW that do not have continuous time limit, and thus cannot be related to the CTQW.

We also found from theorem \ref{thm:fixedCoinTheorem} that certain coins do not need to vary in the continuum limit, as long as they are roots of unity. Finally, we concluded from theorem \ref{thm:ctcs} that various types of Dirac equations can be obtained depending on how the continuum limit of the DTQW is taken. Space and time limits taken simultaneously yield different answers than when time is taken followed by space.
%\textbf{Restate main conclusions from theorems here.}
%\subsection{Implications}
%The main implication of theorem \ref{thm:ctdsUnitaries} is that the continuous time limit exists for any coin of the form $C=e^{-i(\psi_0)\sigma_z/2}\sigma_y e^{-i\theta_1\Delta t\sigma_y/2} e^{-i(\phi_0)\sigma_z/2}$, granted the limit is taken 2 steps at a time. Theorem \ref{thm:connection} builds off of this theorem and an idea presented in Ref.~\onlinecite{Strauch} by claiming that the continuous time limit of the DTQW can be always be related to the CTQW if the coin is of the form $\exp-\frac{i\theta}{2}(\sigma_y\cos\psi_0-\sigma_x\sin\psi_0)$, where $\theta\to p\pi+\theta_1\Delta t$ in the continuum limit for some odd integer $p$ and $\theta_1\in\mathbb{R}$. These two implications are interesting because they show that there exists unitary matrices used as coins in the DTQW that do not have continuous time limits, and thus cannot be related to the CTQW. 

%The implications of theorem \ref{thm:ctcs} is that it shows that various types of Dirac equations can be obtained depending on how the continuum limit of the DTQW is taken. Space and time limits taken simultaneously yields different answers than time and then space. Lastly, the main point of theorem \ref{thm:fixedCoinTheorem} was to demonstrate that certain coins do not need to vary in the continuum limit, as long as they are roots of unity.

\subsection{Open Questions}
There are many open questions pertaining to these ideas. Which types of coins have continuum limits in higher spatial dimensions? What do the connections between the DTQW and CTQW look like in higher spatial dimensions? What would the analogous theorems look like if we introduced multiple coins?

The connections between various continuum limits of the DTQW to the  massless/massive Dirac equation shown in this work and in others suggest the possibility of a new universal quantum computational architecture involving the scattering of particles obeying the Dirac or Schrodinger equations. Are these connections the most that can be made on the topic of computation, or is there something more? Stated another way, do quantum walks involved in quantum computational algorithms have continuum limits which can be related to the Dirac or Schrodinger equation? If they do, would it imply there is a way to utilize the Dirac or Schrodinger dynamics to obtain the results of quantum walk algorithms? The results and techniques shown in this work would certainly help obtain such an answer.

\begin{acknowledgments}
We wish to acknowledge the support of Boston University, as well as very constructive discussions with Tamiro Villazon, Pieter Claeys, Chonkit Pun, Pranay Patil, Parker Kuklinski, and Chris Laumann.
\end{acknowledgments}

\newpage

%%%%%%%%%%%%%%%%%%%%%%%%%%%%%%%%%%%%%%%%%%%%%%%%%%%%%%%%%%%%%%%%%%%%%%%%%%%%%%%%%%%%%%%%%%%%%%%%%%%%%%%%%%%%%%%%%%%%%%%%%%%%%%

\appendix
\section{Proof of General Continuous Time Limit of the DTQW (Theorem \ref{thm:ctdsUnitaries})}\label{app:ctdsUnitaries}
Before we begin, let's reiterate the theorem we wish to prove:

\ctdsU*

We begin by stating the continuous time limit of the DTQW for coin $C$ and skipping $n$ steps, from definition \ref{def:CLimit}:
\begin{equation}\label{eq:ham}
i\partial_t\vv{\psi}(x,t)=\lim_{\Delta t\to 0}\frac{(SC)^n-\mathbb{I}}{n\Delta t}\vv{\psi}(x,t)
\end{equation}

Now we construct lemmas to prove Theorem \ref{thm:ctdsUnitaries}. Our first lemma will be an algebraic expansion of $(\widetilde{S}C)^n$ that will help make manifest later lemmas, where $\widetilde{S}$ is the Fourier transform of $S$, which is defined by $\widetilde{S}=e^{ ik\Delta x\sigma_z}$.

%Lemma algebraically expanding (SC)^n%
\begin{lemma}\label{lma:algebra}
Let $\psi'_0=\psi_0-2k\Delta x$, $A=R_z(\psi'_0)R_y(\theta_0)R_z(\phi_0)$, $B=\psi_1\sigma_z A+\theta_1 \sigma_y R_z(-2\psi'_0)A+\phi_1 A \sigma_z$, and $\widetilde{S}$ be the Fourier transform of $S$. Then the following is true up to $O(\Delta t)$:
\begin{equation} \label{eq:alg}
    (\widetilde{S}C)^n=(e^{i\delta}A)^n(1-\frac{i \Delta t}{2}A^{-1}\sum_{j=0}^{n-1}A^{-j}BA^j)
\end{equation}
\end{lemma}

\begin{proof}
After substituting $\psi,\phi,$ and $\theta$ in terms of $\phi_0,\psi_0,\theta_0,\phi_1,\psi_1,$ and $\theta_1$, the rotation matrices in Equation (\ref{eq:rotations}) become $R_z(\psi)=R_z(\psi_0)(1-\frac{i\psi_1\Delta t}{2}\sigma_z+(\Delta t^2))$ and so on for $R_z(\phi)$ and $R_y(\theta)$. After doing this substitution and going to Fourier space (so $S\rightarrow\widetilde{S}=R_z(-2k\Delta x)$), we get the following for $(\widetilde{S}C)^n$:
\begin{align}\label{eq:sc^n}
    (\widetilde{S}C)^n&=e^{i\delta n}(R_z(-2k\Delta x)R_z(\psi_0)(1-\frac{i\psi_1\Delta t}{2}\sigma_z+O(\Delta t^2))\\
    &\times R_y(\theta_0)(1-\frac{i\theta_1\Delta t}{2}\sigma_y+O(\Delta t^2))R_z(\phi_0)(1-\frac{i\phi_1\Delta t}{2}\sigma_z+O(\Delta t^2)))^n
\end{align}
We now make the substitution $\psi'_0=\psi_0-2k$ and expand Equation (\ref{eq:sc^n}) further:
\begin{equation}\label{eq:reduced1sc^n}
    \begin{split}
        (\widetilde{S}C)^n&=e^{i\delta n} [R_z(\psi'_0)R_y(\theta_0)R_z(\phi_0)-\frac{i \Delta t}{2}(\psi_1 \sigma_z R_z(\psi'_0)R_y(\theta_0)R_z(\phi_0)\\
        &+\theta_1 y R_z(-2\psi_0')R_z(\psi'_0)R_y(\theta_0)R_z(\phi_0) +\phi_1 R_z(\psi'_0)R_y(\theta_0)R_z(\phi_0) \sigma_z)+O(\Delta t^2)]^n\\
        &=e^{i\delta n}(A-\frac{i \Delta t}{2} B+O(\Delta t^2))^n\\
        &=e^{i\delta n}(A^n-\frac{i \Delta t}{2}(A^{n-1}B+A^{n-2}BA+...+ABA^{n-2}+BA^{n-1})+O(\Delta t^2))\\
        &=e^{i\delta n}(A^n-\frac{i \Delta t}{2}\sum_{j=0}^{n-1}A^{n-1-j}BA^j+O(\Delta t^2))\\
        &=(e^{i\delta}A)^n(1-\frac{i \Delta t}{2}A^{-1}\sum_{j=0}^{n-1}A^{-j}BA^j+O(\Delta t^2))
    \end{split}
\end{equation}
\end{proof}
Now we make a statement concerning the $O(\Delta t^2)$ terms:
\begin{lemma}
The continuous time limit as defined in Equation \ref{eq:ham} will be independent of any $O(\Delta t^2)$ terms in the parameters $\psi$, $\theta$, and $\phi$.
\end{lemma}
\begin{proof}
Examining the last line of Equation \ref{eq:reduced1sc^n}, we see that the only contribution of the $O(\Delta t^2)$ terms in the parameters $\psi$, $\theta$, and $\phi$ will be in the $O(\Delta t^2)$ term. The $O(\Delta t^2)$ term in the last line of Equation \ref{eq:reduced1sc^n} does not contribute to the continuous time limit defined in equation \ref{eq:ham} because it goes to zero as the limit is taken. Thus, the $O(\Delta t^2)$ terms in the parameters $\psi$, $\theta$, and $\phi$ do not contribute to the continuous time limit.
\end{proof}
The next lemma uses lemma \ref{lma:algebra} to constrain the values $n$ can take for a finite limit in Equation (\ref{eq:ham}) to exist.

%Lemma forbidding n=1%
\begin{lemma}\label{lma:n=1}
There is no continuous time limit as defined in Equation (\ref{eq:ham}) for n=1.
\end{lemma}

\begin{proof}
For the Hamiltonian in Equation (\ref{eq:ham}) to be finite, $(SC)^n$ must equal $\mathbb{I}+O(\Delta t)$, and thus $\widetilde{S}C$ must equal $\mathbb{I}+O(\Delta t)$ as well. Therefore, from Equation (\ref{eq:alg}), $(e^{i\delta}A)^n$ must equal identity if $\widetilde{S}C=\mathbb{I}+O(\Delta t)$. The only unitary operator $e^{i\delta}A$ that could possibly satisfy $(e^{i\delta}A)^n=\mathbb{I}$ for $n=1$ is the identity operator itself, but $e^{i\delta}A$ cannot even equal identity, as $A$ has $k$ dependence from containing $\widetilde{S}$, and the angles are not permitted to depend on $k$, so there is no possible way to cancel out the $k$ dependence. Thus, there is no continuous time limit defined in Equation (\ref{eq:ham}) for $n=1$.
\end{proof}

Next we use the reasoning from lemma \ref{lma:n=1} to further constrain the values $\theta_0$ and $\delta$ can take.

\begin{lemma}\label{lma:angles}
For the limit defined in Equation (\ref{eq:ham}) to be finite, $\theta_0$ and $\delta$ must be constrained such that $\theta_0=p\pi$ and $\delta=\frac{2\pi l}{n}-p\pi$ for odd integer $p$ and any integer $l$.
\end{lemma}

\begin{proof}
Following up on the constraint that $(e^{i\delta}A)^n=\mathbb{I}$ from lemma \ref{lma:n=1}, let $U$ be the diagonalization matrix of $A$, and let $D$ be the matrix of eigenvalues of $A$. Then we have the following:
\begin{align}
    &(e^{i\delta}A)^n=e^{in\delta}(U^{-1}DUU^{-1}DUU^{-1}DU\ldots)=e^{in\delta}U^{-1}D^nU=\mathbb{I}\\
    &\rightarrow e^{in\delta}D^n=
    UU^{-1}=\mathbb{I}\rightarrow e^{in\delta}D^n=\mathbb{I}
\end{align}
 so if we set the eigenvalues of $e^{i\delta}A$ equal to an $n^{th}$ root of unity $e^{2\pi l/n}$ where $l=0,1,2,...$ (which is equivalent to the constraint $(e^{i\delta}A)^n=\mathbb{I}$), we recover the following constraint for $\theta_0$:
\begin{equation}
    \cos{\theta_0/2}=\frac{\cos{(\frac{2\pi l}{n}-\delta})}{\cos{(\frac{\psi'_0+\phi_0}{2}})}
\end{equation}
Because none of the angles have $k$ dependence, the only way this condition can hold true is if $\cos{\theta_0/2}=0$ or $\theta_0=p\pi$ where $p=1,3,5,...$. This also gives a constraint on $\delta$, being $\delta=\frac{2\pi l}{n}-\frac{p\pi}{2}$. As a remark, the reason why the choice of overall phase is important here is that it shifts the zero point energy of the Hamiltonian in question and will make the dependence on other variables more manifest in the continuum limit (physical quantities are the differences in energies/eigenvalues of a Hamiltonian, not the eigenvalues themselves).
\end{proof}

The next lemma uses the constraints on $\theta_0$ and $\delta$ from lemma \ref{lma:angles} to impose a constraint on $n$.

%Lemma constraining n%
\begin{lemma}\label{lma:n}
For the limit defined in Equation (\ref{eq:ham}) to be finite, $n$ must be even.
\end{lemma}

\begin{proof}
Substituting our constraint for $\theta_0$ from lemma \ref{lma:angles} into $A$, we get the following:
\begin{equation}\label{eq:A}
    A=-i\sin{\frac{p\pi}{2}}R_z(\psi'_0-\phi_0)\sigma_y    
\end{equation}
Now consider $n$ even. Substituting this form of $A$ and our constraint on $\delta$ from lemma \ref{lma:angles} into $(e^{i\delta}A)^{n}$ in the last line of equation \ref{eq:reduced1sc^n}, where $n=2w$ for some integer $w$, we find that $(e^{i\delta}A)^{2w}=(-e^{2i\delta}\mathbb{I})^w=\mathbb{I}$, as $A^2=-\mathbb{I}$ and $(-e^{2i\delta})^w=-\mathbb{I}$ for all $w$. This implies that even powers of $n$ will satisfy $(e^{i\delta}A)^n=\mathbb{I}$. As for odd $n$, we can write $n=2m+1$ for some integer $m$ to obtain the following:
\begin{equation}
    (e^{i\delta}A)^n=(e^{i\delta}A)^{2m+1}=e^{i\delta}A
\end{equation}
This cannot equate to identity, as we showed in lemma \ref{lma:n=1} that for $n=1$ no parameterization of $A$ can make $e^{i\delta}A=\mathbb{I}$. Thus, n must be even to have a finite continuum limit as defined in Equation (\ref{eq:ham}).
\end{proof}

Because the constraints on $n$ and $\theta_0$ hold true for all $l$ from the last two lemmas, we will choose $l=0$ for the remainder of the proof without loss of generality.

%Our next lemma will introduce a simplified form of $\delta$.

%\begin{lemma}\label{lma:reducedDelta}
%$\delta=-\frac{p\pi}{2}$ will allow the constraint %$(e^{i\delta}A)^n=\mathbb{I}$ to be satisfied for all %even $n$.
%\end{lemma}

%\begin{proof}
%We have that $\delta=-\frac{p\pi}{2}$ is simply $\delta=\frac{2\pi l}{n}-\frac{p\pi}{2}$ with $l=0$. Using the constrained form of $A$ found in Equation (\ref{eq:A}), we have that $(e^{i\delta}A)^n=\mathbb{I}$ for $\delta=-\frac{p\pi}{2}$ and even $n$. Thus, we can just use $\delta=-\frac{p\pi}{2}$ from now on.
%\end{proof}

Now we plug in the constraints from lemmas \ref{lma:angles} and \ref{lma:n} to obtain the final forms of $C$, $\widetilde{S}C$, $(\widetilde{S}C)^n$, and most importantly $H$.

%Final Lemma containing C, SC, (SC)^n, and H%
\begin{lemma}\label{lma:H}
Equation (\ref{eq:ham}) will have a finite limit if $C$, $\widetilde{S}C$, $(\widetilde{S}C)^n$, and $H$ are the following, where $S$ is the shift operator defined in equation
\ref{eq:defn}:
\begin{equation}\label{eq:finalH}
    \begin{split}
    &C=(-1)^{p\pi}(R_z(\psi_0-\phi_0)\sigma_y-\frac{i\Delta t}{2}[i(\phi_1-\psi_1)R_z(\psi_0-\phi_0)\sigma_x+\theta_1 R_z(\psi_0-\phi_0)])\\
    &\widetilde{S}C=(-1)^{p\pi}(R_z(\psi'_0-\phi_0)\sigma_y-\frac{i\Delta t}{2}[i(\phi_1-\psi_1)R_z(\psi'_0-\phi_0)\sigma_x+\theta_1 R_z(\psi'_0-\phi_0)])\\
    &(\widetilde{S}C)^n=1-\frac{in\theta_1\Delta t}{4}(R_z(-2\phi_0)+R_z(2\psi'_0))\sigma_y\\
    &H=-\frac{\theta_1}{4}(R_z(-2\phi_0)+S^2 R_z(2\psi_0))\sigma_y.
    \end{split}
\end{equation}
\end{lemma}

\begin{proof}
To find $\widetilde{S}C$, we take the $n_{th}$ root of both sides of the third line of Equation (\ref{eq:reduced1sc^n}). This will yield the following:
\begin{equation}
    \widetilde{S}C=e^{i\delta n}(A-\frac{i \Delta t}{2} B).
\end{equation}
Plugging in for the constrained versions of $A$ and $B$ and reducing, we obtain the following:
\begin{equation}\label{eq:sc}
    \widetilde{S}C=(-1)^{p\pi}(R_z(\psi'_0-\phi_0)\sigma_y-\frac{i\Delta t}{2}[i(\phi_1-\psi_1)R_z(\psi'_0-\phi_0)\sigma_x+\theta_1 R_z(\psi'_0-\phi_0)])
\end{equation}
To find $C$, we simply multiply Equation (\ref{eq:sc}) by $\widetilde{S}^{-1}=R_z(2k)$, obtaining the following:
\begin{equation}
    C=(-1)^{p\pi}(R_z(\psi_0-\phi_0)\sigma_y-\frac{i\Delta t}{2}[i(\phi_1-\psi_1)R_z(\psi_0-\phi_0)\sigma_x+\theta_1 R_z(\psi_0-\phi_0)])
\end{equation}
Next we will find $(\widetilde{S}C)^n$ by evaluating the sum in the last line of Equation (\ref{eq:reduced1sc^n}) using the constrained form of $A$ in Equation (\ref{eq:A}). One can show that $A^2=-1$ and $A^{-1}=-A$, so the sum in Equation (\ref{eq:reduced1sc^n}) becomes the following:
\begin{equation}
    A^{-1}\sum_{j=0}^{n-1}A^{-j}BA^j=-A\sum_{j=0}^{n-1}(-1)^j A^{j}BA^j
\end{equation}
Now we split up the sum into even and odd terms:
\begin{equation}
    \begin{split}
        -A\sum_{j=0}^{n-1}(-1)^j A^{j}BA^j&=-A(\sum_{j=odds}^{n-1}(-1)^j A^{j}BA^j+\sum_{j=evens}^{n-2}(-1)^j A^{j}BA^j)\\
        &=-A\frac{n}{2}(-ABA+B)=-\frac{n}{2}\{A,B\}\\
        &=\frac{n\theta_1}{2}(R_z(-2\phi_0)+R_z(2\psi'_0))\sigma_y
    \end{split}
\end{equation}
We used Equation (\ref{eq:A}) in the last line, so Equation (\ref{eq:reduced1sc^n}) becomes the following:
\begin{equation}\label{eq:sc^nfinal}
    (SC)^n=1-\frac{in\theta_1\Delta t}{4}(R_z(-2\phi_0)+R_z(2\psi'_0))\sigma_y
\end{equation}
Now we can find $\widetilde{H}$ by evaluating the limit in Equation (\ref{eq:ham}) using Equation (\ref{eq:sc^nfinal}) to obtain the following Hamiltonian in Fourier space:
\begin{equation}\label{eq:FourierHam}
   \widetilde{H}=-\frac{\theta_1}{4}(R_z(-2\phi_0)+R_z(2\psi'_0))\sigma_y
\end{equation}
Fourier transforming back and resubstituting for $\psi'_0$, we obtain the following $H$:
\begin{equation}
    H=-\frac{\theta_1}{4}(R_z(-2\phi_0)+S^2 R_z(2\psi_0))\sigma_y
\end{equation}
where $S$ is the shift operator defined in Equation (\ref{eq:defn}).
\end{proof}
Combining lemmas \ref{lma:angles}, \ref{lma:n}, and \ref{lma:H}, we prove Theorem \ref{thm:ctdsUnitaries}.

\newpage
%%%%%%%%%%%%%%%%%%%%%%%%%%%%%%%%%%%%%%%%%%%%%%%%%%%%%%%%%%%%%%%%%%%%%%%%%%
%%%%%%%%%%%%%%%%%%%%%%%%%%%%%%%%%%%%%%%%%%%%%%%%%%%%
\section{Proof of Time Evolution from Continuous Time Limit of DTQW Hamiltonian (Corollary \ref{corr:ctdsWavefunction})}\label{app:corrctdsWavefunction}

We start by reiterating the corollary:

\ctdsWaveFunctions*

We begin the theorem by finding the eigenvalues and eigenvectors of the hamiltonian in fourier space:

\begin{lemma}\label{lma:eigs}
The eigenvalues of the Hamiltonian in fourier space are $\pm\lambda(k)=\pm\cos{(k-\frac{\phi_0+\psi_0}{2})}$ with corresponding eigenvectors $\vv{\pm\lambda}=\frac{1}{\sqrt{2}}\begin{pmatrix}\pm ie^{i(k+\frac{\phi_0+\psi_0}{2})}\\1\end{pmatrix}$.
\end{lemma}
\begin{proof}
The hamiltonian written in fourier space is the following:
\begin{equation}
    \widetilde{H}(k)=-\frac{\theta_1}{4}(R_z(-2\phi_0)+e^{2ik\Delta x}R_z(2\psi_0))\sigma_y
\end{equation}
It follows from straightforward eigenvalue decomposition that the eigenvalues and eigenvectors are those in lemma \ref{lma:eigs}
\end{proof}

Our next lemma relates the fourier transform of $\vv{\Psi}(x,t)$, denoted  $\widetilde{\vv{\Psi}}(k,t)$, to the fourier transform of the initial conditions of $\vv{\Psi}(x,t)$, denoted $\widetilde{\vv{\Psi}}(k,0)$.

\begin{lemma}
Let $U$ be the unitary diagonalization matrix of eigenvectors of $\widetilde{H}(k)$, so $U=\frac{1}{\sqrt{2}}\begin{pmatrix}ie^{i(k+\frac{\phi_0-\psi_0}{2})}&-ie^{i(k+\frac{\phi_0-\psi_0}{2})}\\1 &1\end{pmatrix}$. Then $\widetilde{\vv{\Psi}}(k,t)=U e^{-i\lambda(k)\sigma_zt}U^\dagger\widetilde{\vv{\Psi}}(k,0)$
\end{lemma}

\begin{proof}
If $U$ is the diagonalization matrix of eigenvectors of $\widetilde{H}(k)$, we can write $U^\dagger\widetilde{H}(k)U=\lambda\sigma_z$. Therefore, because $UU^\dagger=\mathbb{I}$ by unitarity, we have the following:
\begin{align}
    \widetilde{\vv{\Psi}}(k,t)&=e^{-i\widetilde{H}(k)t}\widetilde{\vv{\Psi}}(k,0)=UU^\dagger e^{-i\widetilde{H}(k)t}UU^\dagger \widetilde{\vv{\Psi}}(k,0)\\
    &=Ue^{iU^\dagger\widetilde{H}(k)Ut}U^\dagger \widetilde{\vv{\Psi}}(k,0)=Ue^{-i\lambda(k)\sigma_zt}U^\dagger \widetilde{\vv{\Psi}}(k,0)
\end{align}
\end{proof}
Our next lemma recovers the explicit expression for $Ue^{-i\lambda(k)\sigma_zt}U^\dagger \widetilde{\vv{\Psi}}(k,0)$:
\begin{lemma}\label{lma:big}
Let $U$ and $\lambda(k)$ be defined as in lemma \ref{lma:eigs}. Then we have the following expression for $Ue^{-i\lambda(k)\sigma_zt}U^\dagger \widetilde{\vv{\Psi}}(k,0)$, where $\alpha=\frac{\phi_0+\psi_0}{2}$ and $\beta=\frac{\phi_0-\psi_0}{2}$:
\begin{align}
&Ue^{-i\lambda(k)\sigma_zt}U^\dagger\widetilde{\vv{\Psi}}(k,0)=\\
&\frac{1}{2}\begin{pmatrix}(e^{\frac{i\theta_1t}{2}\cos{(k-\alpha)}}+e^{-\frac{i\theta_1t}{2}\cos{(k-\alpha)}})\widetilde{\Psi}_L(k,0)+i(e^{\frac{i\theta_1t}{2}\cos{(k-\alpha)}}-e^{-\frac{i\theta_1t}{2}\cos{(k-\alpha)}})e^{i(k+\beta)}\widetilde{\Psi}_R(k,0)\\-i(e^{\frac{i\theta_1t}{2}\cos{(k-\alpha)}}-e^{-\frac{i\theta_1t}{2}\cos{(k-\alpha)}})e^{-i(k+\beta)}\widetilde{\Psi}_L(k,0)+(e^{\frac{i\theta_1t}{2}\cos{(k-\alpha)}}+e^{-\frac{i\theta_1t}{2}\cos{(k-\alpha)}})\widetilde{\Psi}_R(k,0)\end{pmatrix}
\end{align}
\end{lemma}
We obtain lemma \ref{lma:big} through straightforward matrix multiplication. Our next lemmas will introduce some integrals and convolutions that we will need when computing the inverse fourier transform of the equation in lemma \ref{lma:big}.
\begin{lemma}\label{lma:Bessel}
Let $\mathcal{F}^{-1}$ denote the inverse fourier transform, and $*$ denote the convolution. Then we have the following inverse fourier transforms, where $J_n(t)$ is the $n^{\text{th}}$ order Bessel function of the first kind:
\begin{align}
    &\mathcal{F}^{-1}(e^{\pm \frac{i\theta_1t}{2}\cos{(k-\alpha)}}\widetilde{\Psi}_{L,R}(k,0))=\mathcal{F}^{-1}(e^{\pm \frac{i\theta_1t}{2}\cos{\alpha}\cos{k}}e^{\pm \frac{i\theta_1t}{2}\sin{\alpha}\sin{k}}\widetilde{\Psi}_{L,R}(k,0))\\
    &=\mathcal{F}^{-1}(e^{\pm \frac{i\theta_1t}{2}\cos{\alpha}\cos{k}})*\mathcal{F}^{-1}(e^{\pm\frac{i\theta_1t}{2}\sin{\alpha}\sin{k}})*\mathcal{F}^{-1}(\widetilde{\Psi}_{L,R}(k,0))\\
    &=\sum_{n=-\infty}^\infty(\pm i)^{m-n}J_{m-n}(\frac{\theta_1t}{2})e^{i\alpha (m-n)}\Psi_{L,R}(n\Delta x,0)
\end{align}
\end{lemma}
\begin{proof}
The first equality of the equation in lemma \ref{lma:Bessel} is true by elementary trigonometric identities, and the second line is true by the convolution theorem. For the third line, we need the following inverse fourier transforms
\begin{align}
    &\mathcal{F}^{-1}(e^{\pm \frac{i\theta_1t}{2}\cos{\alpha}\cos{k}})=\frac{\Delta x}{2\pi}\int^{\frac{\pi}{\Delta x}}_{-\frac{\pi}{\Delta x}}dke^{ikn\Delta x}e^{\pm \frac{i\theta_1t}{2}\cos{\alpha}\cos{k}}=(\pm i)^nJ_n(\frac{\theta_1t}{2}\cos{\alpha})\\
    &\mathcal{F}^{-1}(e^{\pm \frac{i\theta_1t}{2}\sin{\alpha}\sin{k}})=\frac{\Delta x}{2\pi}\int^{\frac{\pi}{\Delta x}}_{-\frac{\pi}{\Delta x}}dke^{ikn\Delta x}e^{\pm \frac{i\theta_1t}{2}\sin{\alpha}\sin{k}}=(\mp 1)^nJ_n(\frac{\theta_1t}{2}\sin{\alpha})\\
    &\mathcal{F}^{-1}(\widetilde{\Psi}_{L,R}(k,0))=\Psi_{L,R}(m\Delta x,0)
\end{align}
Now we find $\mathcal{F}^{-1}(e^{\pm\frac{i\theta_1t}{2}\cos{\alpha}\cos{k}})*\mathcal{F}^{-1}(e^{\pm\frac{i\theta_1t}{2}\sin{\alpha}\sin{k}})$:
\begin{align}
    &\mathcal{F}^{-1}(e^{\pm \frac{i\theta_1t}{2}\cos{\alpha}\cos{k}})*\mathcal{F}^{-1}(e^{\pm\frac{i\theta_1t}{2}\sin{\alpha}\sin{k}})\\
    &=(\mp 1)^{n}\sum_{j=-\infty}^\infty(-i)^{j}J_{j}(\frac{\theta_1t}{2}\cos{\alpha})J_{n-j}(\frac{\theta_1t}{2}\sin{\alpha})\\
    &=(\pm i)^{n}J_{n}(\frac{\theta_1t}{2})e^{i\alpha n},
\end{align}
where in the last line we used one of Graf’s and Gegenbauer’s addition theorems (Ref.~\onlinecite{NIST:DLMF} Eq. 10.23.7).  Now we convolve this with $\Psi_{L,R}(m\Delta x,0)$:
\begin{align}
    &(\pm i)^{n}J_{n}(\frac{\theta_1t}{2})e^{i\alpha n}*\Psi_{L,R}(m\Delta x,0)\\
    &=\sum_{n=-\infty}^\infty(\pm i)^{m-n}J_{m-n}(\frac{\theta_1t}{2})e^{i\alpha (m-n)}\Psi_{L,R}(n\Delta x,0)
\end{align}
\end{proof}

Next we have our last lemma:

\begin{lemma}\label{lma:inverseFourier}
Given the expression for $Ue^{-i\lambda(k)\sigma_zt}U^\dagger \widetilde{\vv{\Psi}}(k,0)$ in lemma \ref{lma:big}, we have the following:
\begin{align}
&\mathcal{F}^{-1}(Ue^{-i\lambda(k)\sigma_zt}U^\dagger \widetilde{\vv{\Psi}}(k,0))=\vv{\Psi}(m\Delta x,t)\\
&=\frac{1}{2}\sum_{n=-\infty}^\infty i^{m-n}e^{i\alpha(m-n)}J_{m-n}(\frac{\theta_1t}{2})\begin{pmatrix}(1+(-1)^{m-n})\Psi_L(n\Delta x,0)+ie^{i\beta}(1-(-1)^{m-n})\Psi_R((n+1)\Delta x,0) \\-ie^{-i\beta}(1-(-1)^{m-n})\Psi_L((n-1)\Delta x,0)+(1+(-1)^{m-n})\Psi_R(n\Delta x,0) \end{pmatrix}
\end{align}
\end{lemma}

\begin{proof}
Observing the expression for $Ue^{-i\lambda(k)\sigma_zt}U^\dagger \widetilde{\vv{\Psi}}(k,0)$ in lemma \ref{lma:big}, we use lemma \ref{lma:Bessel} to go through each term and calculate the convolution.
\end{proof}
Thus, lemma \ref{lma:inverseFourier} recovers the time evolution equation in position space for the hamiltonian from theorem \ref{thm:ctdsUnitaries}. 
%%%%%%%%%%%%%%%%%%%%%%%%%%%%%%%%%%%%%%%%%%%%%%%%%%%%%%%%%%%%%%%%%%%%%%%%%%%%%%%%%%%%%%%%%%%%%%%%%%%%%%%%%%%%%%%%%%%%%%%%%%%%%%

\newpage
\section{Proof of Continuous Space-Time Limit with no Coin Variation (Theorem \ref{thm:fixedCoinTheorem})}\label{app:FixedCoinTheorem}
We begin by restating the theorem:

\noVariation*

To prove Theorem \ref{thm:fixedCoinTheorem} we will construct lemmas as was done in section \ref{section:CTLimit}. We will use theorem \ref{theorem:rootofunity} from section \ref{section:rootsofunity} to prove theorem \ref{thm:fixedCoinTheorem}. First we prove that the coin must be of the form of Equation (\ref{eq:ctcsCoin}) by considering the following general unitary coin, where again $|\hat{n}|=\sqrt{n_x^2+n_y^2+n_z^2}=1$:
\begin{equation}\label{eq:fixedC}
    C=e^{i\delta}R_n(\theta)=e^{i\delta}\exp-i\theta \hat{n}\cdot\vv{\sigma}/2
\end{equation}

\begin{lemma}\label{lma:fixedCoinHam}
Let $C$ be a general unitary operator as defined in Equation (\ref{eq:fixedC}). For the continuous space-time limit to exist for this coin, it must be of the form in Equation (\ref{eq:ctcsCoin}).
\end{lemma}

\begin{proof}
The only coins that will have a continuous space-time limit will be those that possess the property such that for some integer $m$, $C^m=1$, as stated in theorem \ref{theorem:rootofunity}. Constraining this property onto the coins in Equation (\ref{eq:fixedC}), we get the following:
\begin{equation}
    \begin{split}
        C^m&=e^{im\delta}\exp-im\theta \hat{n}\cdot\vv{\sigma}/2\\
        &=e^{im\delta}(\cos{m\theta/2}-i\hat{n}\cdot\vv{\sigma}\sin{m\theta/2})=1
    \end{split}
\end{equation}
The $\hat{n}\cdot\vv{\sigma}$ must go away, which constrains $\theta$ to satisfy $\theta=\frac{2\pi l}{m}$, where $l=0,1,2,...$. Applying this constraint yields $C^m=e^{im\delta}(-1)^l$, so we must have that $\delta=\frac{\pi l}{m}$. Thus our original coin has become the following:
\begin{equation}
    C=\exp \frac{i\pi l}{m}\exp \frac{-i\pi l}{m}\hat{n}\cdot\vv{\sigma}
\end{equation}
\end{proof}

Now we will be taking a continuous space-time limit of the DTQW with this coin, and we will see what resultant PDE we obtain.\\

\begin{lemma}\label{lma:fixedCoin}
The Hamiltonian for the continuous space-time limit of the DTQW with coin of the form in Equation (\ref{eq:ctcsCoin}) will be the following:
\begin{equation}
    H=-vn_z\hat{n}\cdot\vv{\sigma}\frac{\partial}{\partial x}.
\end{equation}
\end{lemma}

\begin{proof}
We have the following continuous space-time limit time evolution equation, where $H$ is the resulting Hamiltonian or generator of time evolution for $\Psi$:
\begin{equation}\label{eq:GenFixedU}
   H=i\lim_{\Delta t,\Delta x\to0}\frac{(S(\Delta x)C)^m-\mathbb{I}}{m\Delta t}
\end{equation}
Let $\Delta x=v\Delta t$, so both space and time go to the continuum at the same scale. We focus our attention on the Fourier transformation of the operator in the middle of Equation (\ref{eq:GenFixedU}). We have the following:
\begin{equation}
    \begin{split}
    &\widetilde{S}(\Delta x)C)^n-\mathbb{I}=(e^{ikv\Delta t\sigma_z}C)^m-\mathbb{I}\\
    &=((1+ikv\sigma_z\Delta t+O(\Delta t^2))C)^m-\mathbb{I}
    \end{split}
\end{equation}
Next, we can ignore the $O(\Delta t^2)$ terms, as they will be zero in the end. So we obtain the following:
\begin{equation}\label{eq:fourH}
    \begin{split}
        &(\widetilde{S}(\Delta x)C)^n-\mathbb{I}=(C+ikv\sigma_zC\Delta t)^m-\mathbb{I}\\
        &=C^m+ikv\Delta t(C^{m-1}\sigma_zC+C^{m-2}\sigma_zC^2+...+C\sigma_zC^{m-1}+\sigma_zC^m+O(\Delta t^2))-\mathbb{I}\\
        &=ikv\Delta t\sum_{j=0}^{m-1} C^{m-j}\sigma_zC^{j}=ikv\Delta t\sum_{j=0}^{m-1} C^{-j}\sigma_zC^{j}
    \end{split}
\end{equation}
Again, we can ignore the $O(\Delta t^2)$ terms, and we used the fact that $C^m=1$. We can reduce the series in the last expression of \ref{eq:fourH} in the following way, setting $\alpha=\frac{\pi j}{m}$:
\begin{equation}
    \begin{split}
        C^{-j}\sigma_zC^{j}&=\exp(i\alpha\hat{n}\cdot\vv{\sigma})\sigma_z\exp(-i\alpha\hat{n}\cdot\vv{\sigma})\\
        &=(\cos{\alpha}+i\hat{n}\cdot\vv{\sigma}\sin{\alpha})\sigma_z(\cos{\alpha}-i\hat{n}\cdot\vv{\sigma}\sin{\alpha})\\
        &=2\sin{\alpha}(-n_y\cos{\alpha}+n_x n_z \sin{\alpha})\sigma_x+(2n_y n_z \sin^2{\alpha}+n_x\sin{2\alpha})\sigma_y\\
        &+(\cos{2\alpha}+2n_z^2\sin^2{\alpha})\sigma_z
    \end{split}
\end{equation}
The only terms to survive the sum will be those proportional to $\sin^2{\alpha}$ and $\cos^2{\alpha}$. Thus, we recover the sum:
\begin{equation}
    \sum_{j=0}^{m-1} C^{-j}\sigma_zC^{j}=m n_z \hat{n}\cdot\vv{\sigma}
\end{equation}
And thus our Hamiltonian is the following:
\begin{equation}
    \widetilde{H}=ikvn_z~\hat{n}\cdot\vv{\sigma}
\end{equation}
Inverse Fourier transforming, we recover the Hamiltonian:
\begin{equation}
    H=-vn_z\hat{n}\cdot\vv{\sigma}\frac{\partial}{\partial x}
\end{equation}
\end{proof}
Combining lemmas \ref{lma:fixedCoin} and \ref{lma:fixedCoinHam} we obtain Theorem \ref{thm:fixedCoinTheorem}.

%%%%%%%%%%%%%%%%%%%%%%%%%%%%%%%%%%%%%%%%%%%%%%%%%%%%%%%%%%%%%%%%%%%%%%%%%%%%%%%%%%%%%%%%%%%%%%%%%%%%%%%%%%%%%%%%%%%%%%%%%%%%%%

\newpage
\section{Proof of Continuous Time and then Space Limit of DTQW (Theorem \ref{thm:ctcs})}\label{App:CTCS}

What follows is a short proof of Theorem \ref{thm:ctcs}. Here is the theorem for reference:

\ctcs*

We begin with writing the Fourier space Hamiltonian of the continuous time limit of the DTQW, parameterized by $\phi_0$, $\psi_0$, and $\theta_1$ (from Theorem \ref{thm:ctdsUnitaries}):
    \begin{equation}
        H_C=-\frac{\theta_1}{4}(e^{i\phi_0\sigma_z}+e^{2ik\Delta x\sigma_z}e^{-i\psi_0\sigma_z})\sigma_y
    \end{equation}
    As a reminder, the parameters $\phi_0$, $\psi_0$, and $\theta_1$ are real numbers which cannot depend on $k$. If we don't allow any of these parameters to depend on $\Delta x$, we see that $\lim\limits_{\Delta x\to0}H_C=-\frac{\theta_1}{4}(e^{i\phi_0\sigma_z}+e^{-i\psi_0\sigma_z})\sigma_y$, which contains no spatial derivatives, thereby making it trivial. Therefore we must allow for the parameters to depend on $\Delta x$ somehow. By conjecture \ref{conjecture:phipsi}, the only parameter that can depend on $\Delta x$ is $\theta_1$.
    Now for the following lemma:
    \begin{lemma}
    $\theta_1=\frac{\alpha}{\Delta x}$ for some $\alpha\in\mathbb{R}$ in order for $lim_{\Delta_x\to0}H_C$ to contain a spatial derivative.
    \end{lemma}
    
    \begin{proof}
    Spatial derivatives in Fourier space look the following way, where $\mathcal{F}(g(x))$ is the Fourier transform of $g(x)$:
    \begin{equation}
    \lim_{\Delta x\to0}\frac{e^{ik\Delta x}-1}{\Delta x}=ik=\mathcal{F}(\partial_x)
    \end{equation}
    Now we expand $H_C$ for small $\Delta x$:
    \begin{equation}
        H_C=-\frac{\theta_1}{4}(e^{i\phi_0\sigma_z}+e^{-i\psi_0\sigma_z}+2ik\Delta x\sigma_ze^{-i\psi_0\sigma_z}+O(\Delta x^2))\sigma_y
    \end{equation}
    Given that the only parameter that can depend on $\Delta x$ is $\theta_1$, and that the only term which can contain a spatial derivative is the $3^{\text{rd}}$ term, but only if it is divided by $\Delta x$, it must be the case that $\theta_1$ is proportional to $\frac{1}{\Delta x}$.
    \end{proof}
    Henceforth, we will set  $\theta_1=\frac{\alpha}{\Delta x}$ for some $\alpha\in\mathbb{R}$. Now for the next lemma:
    
    \begin{lemma}
    The only parameterizations that will allow $H_C$ to be finite in the limit $\Delta x\to 0$ are those consistent with the constraint $\phi_0+\psi_0=\pi$.
    \end{lemma}
    \begin{proof}
    We rewrite $H_C$ as above, but now we factor out $e^{-i\psi_0\sigma_z}$:
    \begin{equation}
        H_C=-\frac{\alpha}{4\Delta x}e^{-i\psi_0\sigma_z}(e^{2ik\Delta x\sigma_z}+e^{i(\phi_0+\psi_0)\sigma_z})\sigma_y
    \end{equation}
    In order for the term $e^{2ik\Delta x\sigma_z}+e^{i(\phi_0+\psi_0)\sigma_z}$ to look like a spatial derivative, $e^{i(\phi_0+\psi_0)\sigma_z}$ must be proportional to $-$identity, which equates to $\phi_0+\psi_0=\pi$.
    \end{proof}
    For the sake of completeness, these conditions on the parameters of the coin correspond to the following pre-continuous time limit coin:
    \begin{equation}\label{eq:preCTcoin}
        -\exp(i\phi_0\sigma_z/2)\sigma_x\exp(-i\frac{\alpha\Delta t}{2\Delta x}\sigma_y)\exp(-i\phi_0\sigma_z/2)
    \end{equation}
%    work:
%    \begin{align}
%        &\exp(ik\Delta x \sigma_z)\exp(i\phi_0\sigma_z/2)\sigma_x\exp(-i\frac{\alpha\Delta t}{2\Delta x}\sigma_y)\exp(-i\phi_0\sigma_z/2)\\
%        \times&\exp(ik\Delta x\sigma_z)\exp(i\phi_0\sigma_z/2)\sigma_x\exp(-i\frac{\alpha\Delta t}{2\Delta x}\sigma_y)\exp(-i\phi_0\sigma_z/2)\\
%        =&\exp(i\phi_0\sigma_z/2)\exp(ik\Delta x\sigma_z)\sigma_x\exp(-i\frac{\alpha\Delta t}{2\Delta x}\sigma_y)\exp(ik\Delta x\sigma_z)\sigma_x\\
%        \times&\exp(-i\frac{\alpha\Delta t}{2\Delta x}\sigma_y)\exp(-i\phi_0\sigma_z/2)\\
%        =&\exp(i\phi_0\sigma_z/2)\exp(i\frac{\alpha\Delta t}{2\Delta x}(\sigma_y\cos(2k\Delta x)+\sigma_x\sin(2k\Delta x)))\\
%        \times&\exp(-i\frac{\alpha\Delta t}{2\Delta x}\sigma_y)\exp(-i\phi_0\sigma_z/2)
%    \end{align}
    (from Equation (\ref{eq:CTcoin})). Putting these two lemmas together, we get that the only finite continuous space limit $H_C$ can have which contains spatial derivatives is the following:
    \begin{equation}
        \lim_{\Delta x\to 0}H_C=-\frac{\alpha}{2}e^{-i\psi_0\sigma_z}k\sigma_x
    \end{equation}
    This equates to the following time evolution equation in position space for wave function $\vv{\psi}(x,t)$
    \begin{equation}
        i\partial_t\vv{\psi}(x,t)=i\frac{\alpha}{2}e^{-i\psi_0\sigma_z}\sigma_x\partial_x\vv{\psi}(x,t)
    \end{equation}
    which is in the form of a dirac hamiltonian for a massless particle in the $\sigma_x$ basis (and reduces to the familiar form when $\psi_0=0$).
    %If we had taken a simultaneous space-time limit on a $n=2$ (2 step) DTQW with the coin in \ref{eq:preCTcoin}, it can be shown that we would get the same dirac hamiltonian as above. This is a curious result, as taking the same coin which has a single step continuous space-time limit and then take the time and then space limit on it will not in general yield a finite hamiltonian (but maybe it could, if you used the inverse parameterization on $\Delta x$? Investigate further).

%%%%%%%%%%%%%%%%%%%%%%%%%%%%%%%%%%%%%%%%%%%%%%%%%%%%%%%%%%%%%%%%%%%%%%%%%%%%%%%%%%%%%%%%%%%%%%%%%%%%%%%%%%%%%%%%%%%%%%%%%%%%%%

\newpage
\section{Proof of General Coin CTQW-DTQW Relation (Theorem \ref{thm:connection})}\label{app:CTDTrelation}

We begin by restating the theorem:
\ctdtqwRelation*
To begin, we introduce the following lemma:

\begin{lemma}\label{lma:eigsH}
Let $\widetilde{H}=-\frac{\theta_1}{4}(R_z(-2\phi_0)+\widetilde{S}^2 R_z(2\psi_0))\sigma_y$. Then the eigenvalues of $\widetilde{H}$ are $\pm\frac{\theta_1}{2}\cos(k\Delta x-\frac{\phi_0+\psi_0}{2})$.
\end{lemma}
This lemma is obtained from straightforward eigenvalue decomposition of $\widetilde{H}$. Now for our next lemma:
\begin{lemma}
Let $\vv{\Psi}_\pm(x,t)$ be the inverse fourier transform of the eigenvectors of $\widetilde{H}$. The inverse fourier transform of the time evolution equation of the eigenvectors of $\widetilde{H}$ is $i\partial_t\vv{\Psi}_\pm(x,t)=\pm\frac{\theta_1}{4}(e^{-i\alpha}\vv{\Psi}_\pm(x+\Delta x,t)+e^{i\alpha}\vv{\Psi}_\pm(x-\Delta x,t))$
\end{lemma}
\begin{proof}
From lemma \ref{lma:eigsH} we have the following, where $\vv{\widetilde{\Psi}}_\pm(k,t)$ are eigenvectors of $\widetilde{H}$ and $\alpha=\frac{\phi_0+\psi_0}{2}$:
\begin{align}
    i\partial_t\vv{\widetilde{\Psi}}_\pm(k,t)&=\pm\frac{\theta_1}{2}\cos(k\Delta x-\alpha)\vv{\widetilde{\Psi}}_\pm(k,t)\\
    &=\pm\frac{\theta_1}{4}(e^{ik\Delta x}e^{-i\alpha}+e^{-ik\Delta x}e^{i\alpha})\vv{\widetilde{\Psi}}_\pm(k,t).
\end{align}
Inverse fourier transforming this, we get $i\partial_t\vv{\Psi}_\pm(x,t)=\pm\frac{\theta_1}{4}(e^{-i\alpha}\vv{\Psi}_\pm(x+\Delta x,t)+{e^{i\alpha}\vv{\Psi}_\pm(x-\Delta x,t))}$.
\end{proof}
Now for our next lemma:
\begin{lemma}\label{lma:CTLimitDTQWaveFunctionEig}
The wave functions $\vv{\Psi'}_\pm(x,t)=e^{ i(\pm\frac{\theta_1 t}{2}-\alpha\frac{x}{\Delta x})}\vv{\Psi}_\pm(x,t)$ will satisfy the CTQW time evolution equation $i\partial_t\vv{\Psi'}_\pm(x,t)=\pm\frac{\theta_1}{4}\big[\vv{\Psi'}_\pm(x+\Delta x,t)+\vv{\Psi'}_\pm(x-\Delta x,t)-2\vv{\Psi'}_\pm(x,t)]$.
\end{lemma}
\begin{proof}
It can easily be seen that plugging in $\vv{\Psi}_\pm(x,t)=e^{ i(\mp\frac{\theta_1 t}{2}+\alpha\frac{x}{\Delta x})}\vv{\Psi'}_\pm(x,t)$ to $i\partial_t\vv{\Psi}_\pm(x,t)=\pm\frac{\theta_1}{4}(e^{-i\alpha}\vv{\Psi}_\pm(x+\Delta x,t)+e^{i\alpha}\vv{\Psi}_\pm(x+\Delta x,t))$ will yield the CTQW time evolution equation for $\vv{\Psi'}_\pm(x,t)$.
\end{proof}
Now for our last lemma:
\begin{lemma}\label{lma:CTLimitDTQWWaveFunctionFinal}
$\vv{\Psi}(x,t)$ can be written as a superposition of $\vv{\Psi'}_+(x,t)$ and $\vv{\Psi'}_-(x,t)$, which satisfy the CTQW time evolution equation, in the following way:
\begin{equation}
\vv{\Psi}(x,t)=e^{i\alpha\frac{x}{\Delta x}}(e^{-\frac{i\theta_1 t}{2}}\vv{\Psi'}_+(x,t)+e^{\frac{i\theta_1 t}{2}}\vv{\Psi'}_-(x,t))    
\end{equation}
\end{lemma}
\begin{proof}
Let $P_+$ and $P_-$ be projectors onto the $+$ and $-$ eigenvectors of ${H=-\frac{\theta_1}{4}(R_z(-2\phi_0)+S^2 R_z(2\psi_0))\sigma_y}$. Then we can write the following: 
\begin{align}
\vv{\Psi}(x,t)&=(P_++P_-)\vv{\Psi}(x,t)\\
&=P_+\vv{\Psi}(x,t)+P_-\vv{\Psi}(x,t)\\
&=\vv{\Psi}_+(x,t)+\vv{\Psi}_-(x,t).
\end{align}
From lemma \ref{lma:CTLimitDTQWaveFunctionEig}, we plug in $\vv{\Psi}_\pm(x,t)=e^{ i(\mp\frac{\theta_1 t}{2}+\alpha\frac{x}{\Delta x})}\vv{\Psi'}_\pm(x,t)$ and obtain the expression in lemma  \ref{lma:CTLimitDTQWWaveFunctionFinal}.
\end{proof}

%%%%%%%%%%%%%%%%%%%%%%%%%%%%%%%%%%%%%%%%%%%%%%%%%%%%%%%%%%%%%%%%%%%%%%%%%%%%%%%%%%%%%%%%%%%%%%%%%%%%%%%%%%%%%%%%%%%%%%%%%%%%%%

\section{Continuous Space-Time Limit With Coin Variation For $n=1$}\label{App:CSTn=1}
The following will be a reiteration of some of the results from Ref.~\onlinecite{Molfetta}, but there will be an emphasis on relating the continuous space-time limit to the Dirac equation, or ``Dirac-Type'' equations as we will denote them. Let $\vv{\Psi}(x,t)\in L^2(\mathbb{R})\times L^2(\Sigma)$, where $\Sigma$ is the space spanned by $\ket{L}=\begin{pmatrix}1\\0\end{pmatrix}$ and $\ket{R}=\begin{pmatrix}0\\1\end{pmatrix}$, and let $\sigma_z=\begin{pmatrix}1&0\\0&1\end{pmatrix}$ and $\sigma_x=\begin{pmatrix}0&1\\1&0\end{pmatrix}$. Then the following is the Dirac equation in $1$ space and $1$ time dimension ($1+1$):
\begin{equation}\label{eq:DiracEqn}
    i\partial_t\vv{\Psi}(x,t)=(i\sigma_z\partial_x+\sigma_xm)\vv{\Psi}(x,t).
\end{equation}
A ``Dirac-Type'' equation is the following, where $\hat{A}$ and $\hat{B}$ are any $2\times2$ anti-hermitian and hermitian matrices, respectively:
\begin{equation}\label{eq:DiracTypeEqn}
    i\partial_t\vv{\Psi}(x,t)=(i\hat{A}\partial_x+\hat{B})\vv{\Psi}(x,t).
\end{equation}
Equation (\ref{eq:DiracEqn}) can easily be obtained by taking the continuous space-time limit of the DTQW with the coin $C=e^{im\Delta t\sigma_x}$ and the usual shift operator (in Fourier space) $\widetilde{S}=e^{ik\Delta x\sigma_z}$. In the same fashion, Equation (\ref{eq:DiracTypeEqn}) can easily be obtained by taking the continuous space-time limit of the DTQW with the coin $C=e^{i\Delta t\hat{B}}$, but now with a different shift operator $\widetilde{S}=e^{ik\Delta x\hat{A}}$.
\nocite{*}
\bibliography{aipsamp}% Produces the bibliography via BibTeX.

%merlin.mbs apsrev4-1.bst 2010-07-25 4.21a (PWD, AO, DPC) hacked
%Control: key (0)
%Control: author (8) initials jnrlst
%Control: editor formatted (1) identically to author
%Control: production of article title (-1) disabled
%Control: page (0) single
%Control: year (1) truncated
%Control: production of eprint (0) enabled
\providecommand{\noopsort}[1]{}\providecommand{\singleletter}[1]{#1}%
\begin{thebibliography}{24}%
\makeatletter
\providecommand \@ifxundefined [1]{%
 \@ifx{#1\undefined}
}%
\providecommand \@ifnum [1]{%
 \ifnum #1\expandafter \@firstoftwo
 \else \expandafter \@secondoftwo
 \fi
}%
\providecommand \@ifx [1]{%
 \ifx #1\expandafter \@firstoftwo
 \else \expandafter \@secondoftwo
 \fi
}%
\providecommand \natexlab [1]{#1}%
\providecommand \enquote  [1]{``#1''}%
\providecommand \bibnamefont  [1]{#1}%
\providecommand \bibfnamefont [1]{#1}%
\providecommand \citenamefont [1]{#1}%
\providecommand \href@noop [0]{\@secondoftwo}%
\providecommand \href [0]{\begingroup \@sanitize@url \@href}%
\providecommand \@href[1]{\@@startlink{#1}\@@href}%
\providecommand \@@href[1]{\endgroup#1\@@endlink}%
\providecommand \@sanitize@url [0]{\catcode `\\12\catcode `\$12\catcode
  `\&12\catcode `\#12\catcode `\^12\catcode `\_12\catcode `\%12\relax}%
\providecommand \@@startlink[1]{}%
\providecommand \@@endlink[0]{}%
\providecommand \url  [0]{\begingroup\@sanitize@url \@url }%
\providecommand \@url [1]{\endgroup\@href {#1}{\urlprefix }}%
\providecommand \urlprefix  [0]{URL }%
\providecommand \Eprint [0]{\href }%
\providecommand \doibase [0]{http://dx.doi.org/}%
\providecommand \selectlanguage [0]{\@gobble}%
\providecommand \bibinfo  [0]{\@secondoftwo}%
\providecommand \bibfield  [0]{\@secondoftwo}%
\providecommand \translation [1]{[#1]}%
\providecommand \BibitemOpen [0]{}%
\providecommand \bibitemStop [0]{}%
\providecommand \bibitemNoStop [0]{.\EOS\space}%
\providecommand \EOS [0]{\spacefactor3000\relax}%
\providecommand \BibitemShut  [1]{\csname bibitem#1\endcsname}%
\let\auto@bib@innerbib\@empty
%</preamble>
\bibitem [{\citenamefont {{Lovett}}\ \emph
  {et~al.}(2010{\natexlab{a}})\citenamefont {{Lovett}}, \citenamefont
  {{Cooper}}, \citenamefont {{Everitt}}, \citenamefont {{Trevers}},\ and\
  \citenamefont {{Kendon}}}]{Lovett}%
  \BibitemOpen
  \bibfield  {author} {\bibinfo {author} {\bibfnamefont {N.~B.}\ \bibnamefont
  {{Lovett}}}, \bibinfo {author} {\bibfnamefont {S.}~\bibnamefont {{Cooper}}},
  \bibinfo {author} {\bibfnamefont {M.}~\bibnamefont {{Everitt}}}, \bibinfo
  {author} {\bibfnamefont {M.}~\bibnamefont {{Trevers}}}, \ and\ \bibinfo
  {author} {\bibfnamefont {V.}~\bibnamefont {{Kendon}}},\ }\href@noop {} {\
  (\bibinfo {year} {2010}{\natexlab{a}})},\ \Eprint
  {http://arxiv.org/abs/0910.1024} {0910.1024} \BibitemShut {NoStop}%
\bibitem [{\citenamefont {Ambainis}\ \emph {et~al.}(2001)\citenamefont
  {Ambainis}, \citenamefont {Bach}, \citenamefont {Nayak}, \citenamefont
  {Vishwanath},\ and\ \citenamefont {Watrous}}]{AmbainisDTQW}%
  \BibitemOpen
  \bibfield  {author} {\bibinfo {author} {\bibfnamefont {A.}~\bibnamefont
  {Ambainis}}, \bibinfo {author} {\bibfnamefont {E.}~\bibnamefont {Bach}},
  \bibinfo {author} {\bibfnamefont {A.}~\bibnamefont {Nayak}}, \bibinfo
  {author} {\bibfnamefont {A.}~\bibnamefont {Vishwanath}}, \ and\ \bibinfo
  {author} {\bibfnamefont {J.}~\bibnamefont {Watrous}},\ }\href@noop {} {\
  (\bibinfo {year} {2001})}\BibitemShut {NoStop}%
\bibitem [{\citenamefont {Shenvi}\ \emph {et~al.}(2003)\citenamefont {Shenvi},
  \citenamefont {Kempe},\ and\ \citenamefont {Whaley}}]{Shenvi}%
  \BibitemOpen
  \bibfield  {author} {\bibinfo {author} {\bibfnamefont {N.}~\bibnamefont
  {Shenvi}}, \bibinfo {author} {\bibfnamefont {J.}~\bibnamefont {Kempe}}, \
  and\ \bibinfo {author} {\bibfnamefont {K.~B.}\ \bibnamefont {Whaley}},\
  }\href@noop {} {\bibfield  {journal} {\bibinfo  {journal} {Phys. Rev. A}\ }
  (\bibinfo {year} {2003})}\BibitemShut {NoStop}%
\bibitem [{\citenamefont {{Ambainis}}(2003)}]{AmbainisElem}%
  \BibitemOpen
  \bibfield  {author} {\bibinfo {author} {\bibfnamefont {A.}~\bibnamefont
  {{Ambainis}}},\ }\href@noop {} {\bibfield  {journal} {\bibinfo  {journal}
  {eprint arXiv:quant-ph/0311001}\ } (\bibinfo {year} {2003})}\BibitemShut
  {NoStop}%
\bibitem [{\citenamefont {Jordan}(2017)}]{QuantumAlgorithmZoo}%
  \BibitemOpen
  \bibfield  {author} {\bibinfo {author} {\bibfnamefont {S.}~\bibnamefont
  {Jordan}},\ }\href {http://math.nist.gov/quantum/zoo/} {\enquote {\bibinfo
  {title} {Quantum algorithm zoo},}\ } (\bibinfo {year} {2017})\BibitemShut
  {NoStop}%
\bibitem [{\citenamefont {{Strauch}}(2006)}]{Strauch}%
  \BibitemOpen
  \bibfield  {author} {\bibinfo {author} {\bibfnamefont {F.~W.}\ \bibnamefont
  {{Strauch}}},\ }\href@noop {} {\  (\bibinfo {year} {2006})}\BibitemShut
  {NoStop}%
\bibitem [{\citenamefont {Feynman}\ and\ \citenamefont
  {Hibbs}(1965)}]{FeynHibbs}%
  \BibitemOpen
  \bibfield  {author} {\bibinfo {author} {\bibfnamefont {R.~P.}\ \bibnamefont
  {Feynman}}\ and\ \bibinfo {author} {\bibfnamefont {A.~R.}\ \bibnamefont
  {Hibbs}},\ }\href@noop {} {\emph {\bibinfo {title} {Quantum mechanics and
  path integrals R.P. Feynman A.R. Hibbs}}}\ (\bibinfo  {publisher}
  {McGraw-Hill},\ \bibinfo {year} {1965})\BibitemShut {NoStop}%
\bibitem [{\citenamefont {Requardt}(2006)}]{CLDiscreteGeometries}%
  \BibitemOpen
  \bibfield  {author} {\bibinfo {author} {\bibfnamefont {M.}~\bibnamefont
  {Requardt}},\ }\href {\doibase 10.1142/S0219887806001156} {\bibfield
  {journal} {\bibinfo  {journal} {International Journal of Geometric Methods in
  Modern Physics}\ }\textbf {\bibinfo {volume} {3}} (\bibinfo {year} {2006}),\
  10.1142/S0219887806001156}\BibitemShut {NoStop}%
\bibitem [{\citenamefont {Nesterov}\ and\ \citenamefont
  {Mata}(2019)}]{Nesterov}%
  \BibitemOpen
  \bibfield  {author} {\bibinfo {author} {\bibfnamefont {A.~I.}\ \bibnamefont
  {Nesterov}}\ and\ \bibinfo {author} {\bibfnamefont {H.}~\bibnamefont
  {Mata}},\ }\href {\doibase 10.3389/fphy.2019.00032} {\bibfield  {journal}
  {\bibinfo  {journal} {Frontiers in Physics}\ }\textbf {\bibinfo {volume}
  {7}},\ \bibinfo {pages} {32} (\bibinfo {year} {2019})}\BibitemShut {NoStop}%
\bibitem [{\citenamefont {Mlodinow}\ and\ \citenamefont
  {Brun}(2018)}]{Mlodinowl}%
  \BibitemOpen
  \bibfield  {author} {\bibinfo {author} {\bibfnamefont {L.}~\bibnamefont
  {Mlodinow}}\ and\ \bibinfo {author} {\bibfnamefont {T.~A.}\ \bibnamefont
  {Brun}},\ }\href {\doibase 10.1103/PhysRevA.97.042131} {\bibfield  {journal}
  {\bibinfo  {journal} {Phys. Rev.}\ }\textbf {\bibinfo {volume} {A97}},\
  \bibinfo {pages} {042131} (\bibinfo {year} {2018})},\ \Eprint
  {http://arxiv.org/abs/1802.03910} {arXiv:1802.03910 [quant-ph]} \BibitemShut
  {NoStop}%
%%CITATION = ARXIV:1802.03910;%%
\bibitem [{\citenamefont {Brun}\ and\ \citenamefont
  {Mlodinow}(2019)}]{MlodinowlExperiment}%
  \BibitemOpen
  \bibfield  {author} {\bibinfo {author} {\bibfnamefont {T.~A.}\ \bibnamefont
  {Brun}}\ and\ \bibinfo {author} {\bibfnamefont {L.}~\bibnamefont
  {Mlodinow}},\ }\href {\doibase 10.1103/PhysRevD.99.015012} {\bibfield
  {journal} {\bibinfo  {journal} {Phys. Rev.}\ }\textbf {\bibinfo {volume}
  {D99}},\ \bibinfo {pages} {015012} (\bibinfo {year} {2019})},\ \Eprint
  {http://arxiv.org/abs/1802.03911} {arXiv:1802.03911 [quant-ph]} \BibitemShut
  {NoStop}%
%%CITATION = ARXIV:1802.03911;%%
\bibitem [{\citenamefont {Knight}\ \emph {et~al.}(2003)\citenamefont {Knight},
  \citenamefont {Rold\'an},\ and\ \citenamefont {Sipe}}]{Knight}%
  \BibitemOpen
  \bibfield  {author} {\bibinfo {author} {\bibfnamefont {P.~L.}\ \bibnamefont
  {Knight}}, \bibinfo {author} {\bibfnamefont {E.}~\bibnamefont {Rold\'an}}, \
  and\ \bibinfo {author} {\bibfnamefont {J.~E.}\ \bibnamefont {Sipe}},\ }\href
  {\doibase 10.1103/PhysRevA.68.020301} {\bibfield  {journal} {\bibinfo
  {journal} {Phys. Rev. A}\ }\textbf {\bibinfo {volume} {68}},\ \bibinfo
  {pages} {020301} (\bibinfo {year} {2003})}\BibitemShut {NoStop}%
\bibitem [{\citenamefont {Blanchard}\ and\ \citenamefont
  {Hongler}(2004)}]{Blanchard}%
  \BibitemOpen
  \bibfield  {author} {\bibinfo {author} {\bibfnamefont {P.}~\bibnamefont
  {Blanchard}}\ and\ \bibinfo {author} {\bibfnamefont {M.-O.}\ \bibnamefont
  {Hongler}},\ }\href {\doibase 10.1103/PhysRevLett.92.120601} {\bibfield
  {journal} {\bibinfo  {journal} {Phys. Rev. Lett.}\ }\textbf {\bibinfo
  {volume} {92}},\ \bibinfo {pages} {120601} (\bibinfo {year}
  {2004})}\BibitemShut {NoStop}%
\bibitem [{\citenamefont {{Strauch}}(2007)}]{StrauchRelBig}%
  \BibitemOpen
  \bibfield  {author} {\bibinfo {author} {\bibfnamefont {F.~W.}\ \bibnamefont
  {{Strauch}}},\ }\href@noop {} {\bibfield  {journal} {\bibinfo  {journal}
  {Journal of Mathematical Physics}\ } (\bibinfo {year} {2007})}\BibitemShut
  {NoStop}%
\bibitem [{\citenamefont {{Bracken}}\ \emph {et~al.}(2007)\citenamefont
  {{Bracken}}, \citenamefont {{Ellinas}},\ and\ \citenamefont
  {{Smyrnakis}}}]{Bracken}%
  \BibitemOpen
  \bibfield  {author} {\bibinfo {author} {\bibfnamefont {A.~J.}\ \bibnamefont
  {{Bracken}}}, \bibinfo {author} {\bibfnamefont {D.}~\bibnamefont
  {{Ellinas}}}, \ and\ \bibinfo {author} {\bibfnamefont {I.}~\bibnamefont
  {{Smyrnakis}}},\ }\href@noop {} {\bibfield  {journal} {\bibinfo  {journal}
  {Phys. Rev. A}\ } (\bibinfo {year} {2007})}\BibitemShut {NoStop}%
\bibitem [{\citenamefont {{Gerritsma}}\ \emph {et~al.}(2010)\citenamefont
  {{Gerritsma}}, \citenamefont {{Kirchmair}}, \citenamefont {{Z{\"a}hringer}},
  \citenamefont {{Solano}}, \citenamefont {{Blatt}},\ and\ \citenamefont
  {{Roos}}}]{Zitter}%
  \BibitemOpen
  \bibfield  {author} {\bibinfo {author} {\bibfnamefont {R.}~\bibnamefont
  {{Gerritsma}}}, \bibinfo {author} {\bibfnamefont {G.}~\bibnamefont
  {{Kirchmair}}}, \bibinfo {author} {\bibfnamefont {F.}~\bibnamefont
  {{Z{\"a}hringer}}}, \bibinfo {author} {\bibfnamefont {E.}~\bibnamefont
  {{Solano}}}, \bibinfo {author} {\bibfnamefont {R.}~\bibnamefont {{Blatt}}}, \
  and\ \bibinfo {author} {\bibfnamefont {C.~F.}\ \bibnamefont {{Roos}}},\
  }\href@noop {} {\ \textbf {\bibinfo {volume} {463}},\ \bibinfo {pages} {68}
  (\bibinfo {year} {2010})}\BibitemShut {NoStop}%
\bibitem [{\citenamefont {Molfetta}\ and\ \citenamefont
  {Debbasch}(2012)}]{Molfetta}%
  \BibitemOpen
  \bibfield  {author} {\bibinfo {author} {\bibfnamefont {G.}~\bibnamefont
  {Molfetta}}\ and\ \bibinfo {author} {\bibfnamefont {F.}~\bibnamefont
  {Debbasch}},\ }\href@noop {} {\bibfield  {journal} {\bibinfo  {journal}
  {Journal of Mathematical Physics}\ } (\bibinfo {year} {2012})}\BibitemShut
  {NoStop}%
\bibitem [{{\relax DLMF}()}]{NIST:DLMF}%
  \BibitemOpen
  {\relax DLMF},\ \href {http://dlmf.nist.gov/} {\enquote {\bibinfo {title}
  {{\it NIST Digital Library of Mathematical Functions}},}\ }\bibinfo
  {howpublished} {http://dlmf.nist.gov/, Release 1.0.22 of 2019-03-15},\
  \bibinfo {note} {f.~W.~J. Olver, A.~B. {Olde Daalhuis}, D.~W. Lozier, B.~I.
  Schneider, R.~F. Boisvert, C.~W. Clark, B.~R. Miller and B.~V. Saunders,
  eds.}\BibitemShut {Stop}%
\bibitem [{\citenamefont {Childs}\ and\ \citenamefont
  {Goldstone}(2004)}]{ChildsGoldstone}%
  \BibitemOpen
  \bibfield  {author} {\bibinfo {author} {\bibfnamefont {A.~M.}\ \bibnamefont
  {Childs}}\ and\ \bibinfo {author} {\bibfnamefont {J.}~\bibnamefont
  {Goldstone}},\ }\href {\doibase 10.1103/PhysRevA.70.042312} {\bibfield
  {journal} {\bibinfo  {journal} {Phys. Rev. A}\ }\textbf {\bibinfo {volume}
  {70}},\ \bibinfo {pages} {042312} (\bibinfo {year} {2004})}\BibitemShut
  {NoStop}%
\bibitem [{\citenamefont {Strauch}(2006)}]{StrauchRelSmall}%
  \BibitemOpen
  \bibfield  {author} {\bibinfo {author} {\bibfnamefont {F.~W.}\ \bibnamefont
  {Strauch}},\ }\href@noop {} {\bibfield  {journal} {\bibinfo  {journal} {Phys.
  Rev. A}\ } (\bibinfo {year} {2006})}\BibitemShut {NoStop}%
\bibitem [{\citenamefont {{Childs}}(2009)}]{CTQWUniv}%
  \BibitemOpen
  \bibfield  {author} {\bibinfo {author} {\bibfnamefont {A.~M.}\ \bibnamefont
  {{Childs}}},\ }\href@noop {} {\bibfield  {journal} {\bibinfo  {journal}
  {Physical Review Letters}\ } (\bibinfo {year} {2009})}\BibitemShut {NoStop}%
\bibitem [{\citenamefont {{Lovett}}\ \emph
  {et~al.}(2010{\natexlab{b}})\citenamefont {{Lovett}}, \citenamefont
  {{Cooper}}, \citenamefont {{Everitt}}, \citenamefont {{Trevers}},\ and\
  \citenamefont {{Kendon}}}]{DTQWUniv}%
  \BibitemOpen
  \bibfield  {author} {\bibinfo {author} {\bibfnamefont {N.~B.}\ \bibnamefont
  {{Lovett}}}, \bibinfo {author} {\bibfnamefont {S.}~\bibnamefont {{Cooper}}},
  \bibinfo {author} {\bibfnamefont {M.}~\bibnamefont {{Everitt}}}, \bibinfo
  {author} {\bibfnamefont {M.}~\bibnamefont {{Trevers}}}, \ and\ \bibinfo
  {author} {\bibfnamefont {V.}~\bibnamefont {{Kendon}}},\ }\href@noop {}
  {\bibfield  {journal} {\bibinfo  {journal} {Phys. Rev. A}\ } (\bibinfo {year}
  {2010}{\natexlab{b}})}\BibitemShut {NoStop}%
\bibitem [{\citenamefont {{Childs}}\ \emph {et~al.}(2002)\citenamefont
  {{Childs}}, \citenamefont {{Cleve}}, \citenamefont {{Deotto}}, \citenamefont
  {{Farhi}}, \citenamefont {{Gutmann}},\ and\ \citenamefont
  {{Spielman}}}]{CTQWexp}%
  \BibitemOpen
  \bibfield  {author} {\bibinfo {author} {\bibfnamefont {A.~M.}\ \bibnamefont
  {{Childs}}}, \bibinfo {author} {\bibfnamefont {R.}~\bibnamefont {{Cleve}}},
  \bibinfo {author} {\bibfnamefont {E.}~\bibnamefont {{Deotto}}}, \bibinfo
  {author} {\bibfnamefont {E.}~\bibnamefont {{Farhi}}}, \bibinfo {author}
  {\bibfnamefont {S.}~\bibnamefont {{Gutmann}}}, \ and\ \bibinfo {author}
  {\bibfnamefont {D.~A.}\ \bibnamefont {{Spielman}}},\ }\href@noop {}
  {\bibfield  {journal} {\bibinfo  {journal} {eprint arXiv:quant-ph/0209131}\ }
  (\bibinfo {year} {2002})}\BibitemShut {NoStop}%
\bibitem [{\citenamefont {Farhi}\ and\ \citenamefont
  {Gutmann}(1998)}]{FarhiCTQW}%
  \BibitemOpen
  \bibfield  {author} {\bibinfo {author} {\bibfnamefont {E.}~\bibnamefont
  {Farhi}}\ and\ \bibinfo {author} {\bibfnamefont {S.}~\bibnamefont
  {Gutmann}},\ }\href {\doibase 10.1103/PhysRevA.58.915} {\bibfield  {journal}
  {\bibinfo  {journal} {Phys. Rev.}\ }\textbf {\bibinfo {volume} {A58}},\
  \bibinfo {pages} {915} (\bibinfo {year} {1998})},\ \Eprint
  {http://arxiv.org/abs/quant-ph/9706062} {arXiv:quant-ph/9706062 [quant-ph]}
  \BibitemShut {NoStop}%
%%CITATION = QUANT-PH/9706062;%%
\end{thebibliography}%

\end{document}